\documentclass[12pt, draftclsnofoot, onecolumn]{IEEEtran}
\usepackage[margin=2.54cm, lmargin=2.54cm]{geometry}
\IEEEoverridecommandlockouts
\usepackage{cite}
\usepackage{amsmath,amssymb,amsfonts}

\usepackage{color}
\usepackage{amsthm}
\usepackage{mathrsfs}
\usepackage{bm}
\usepackage[ruled,vlined]{algorithm2e}
\usepackage{subfigure}

\usepackage{graphicx}
\usepackage{textcomp}
\usepackage{multirow}
\usepackage{threeparttable}

\usepackage{stfloats}
\graphicspath{{figures/}}

\usepackage{hyperref}

\newtheorem{theorem}{Theorem}
\newtheorem{lemma}{Lemma}
\newtheorem{remark}{Remark}

\usepackage{url}

\makeatletter
\def\url@leostyle{%
	\@ifundefined{selectfont}{\def\UrlFont{\sf}}{\def\UrlFont{\small\ttfamily}}}
\makeatother
\begin{document}

\title{Fast Successive-Cancellation Decoding of Polar Codes with Sequence Nodes}

\author{\IEEEauthorblockN{Yang Lu,~\IEEEmembership{Student Member,~IEEE}, Ming-Min Zhao,~\IEEEmembership{Member,~IEEE}, \\ Ming Lei,~\IEEEmembership{Member,~IEEE}, and Min-Jian Zhao,~\IEEEmembership{Senior Member,~IEEE}}
\thanks{The authors are with the College of Information Science and Electronic Engineering, Zhejiang University, Hangzhou 310027, China (email: \{22031097, zmmblack, lm1029, mjzhao\}@zju.edu.cn).}}

\maketitle
\vspace{-3em}
\begin{abstract} \vspace{-1em}
Due to the sequential nature of the successive-cancellation (SC) algorithm, the decoding of polar codes suffers from significant decoding latencies. Fast SC decoding is able to speed up the SC decoding process, by implementing parallel decoders at the intermediate levels of the SC decoding tree for some special nodes with specific information and frozen bit patterns. To further improve the parallelism of SC decoding, this paper present a new class of \textcolor{black}{special nodes} composed of a sequence of rate one or single-parity-check (SR1/SPC) nodes, which can be easily found especially in high-rate polar code and is able to envelop a wide variety of existing special node types. Then, we analyse the parity constraints caused by the frozen bits in each descendant node, such that the decoding performance of the SR1/SPC node can be preserved once the parity constraints are satisfied. Finally, a generalized fast decoding algorithm is proposed to decode SR1/SPC nodes efficiently, where the corresponding parity constraints are taken into consideration. Simulation results show that the proposed decoding algorithm of the SR1/SPC node can achieve near-ML performance, and the overall decoding latency can be reduced by 43.8\% as compared to the state-of-the-art fast SC decoder.
\end{abstract}
\vspace{-1em}
\begin{IEEEkeywords}\vspace{-1em}
Polar codes, SC, fast SC decoding, ML decoding, parity constraints. 
\end{IEEEkeywords}

\section{Introduction}
\label{section:1}
\IEEEPARstart{P}{olar} codes are a class of error correction codes which are theoretically proved to have capacity-achieving capabilities under binary-input memoryless channels \cite{Arikan2009Channel}. Due to the capacity-achieving error-correction performance and low-complexity successive cancellation (SC) based algorithm, polar codes are adopted in the control channel of the enhanced mobile broadband (eMBB) use case in the latest 5G cellular standard \cite{5Gstandard}. However, there are two main drawbacks associated with the SC decoding algorithm. First, according to the polarization theory \cite{Arikan2009Channel}, polar codes under SC decoding can achieve the channel capacity only when the code length tends toward infinity. As a result, SC decoding falls short in providing a reasonable error-correction performance for practical moderate code lengths. Second, the sequential bit-by-bit nature of SC decoding leads to high decoding latency and low throughput in terms of hardware implementation, which hinders its application in low-latency communication scenarios such as the ultra-reliable low-latency communication (URLLC) \cite{3GPP} use case of 5G. 

The first drawback mentioned above is mainly due to the fact that SC decoding is suboptimal with respect to maximum-likelihood (ML) decoding. {\color{black}To partially compensate for this sub-optimality, SC list (SCL) decoding algorithm and its LLR-based SCL decoder were presented in \cite{Tal2015List} and \cite{Balatsoukas2015LLR}, respectively. By maintaining a list of candidate codewords, the SCL decoder is able to reduce the performance gap between SC and ML decoding at the cost of increased implementation complexity \cite{Ercan2017error}.} By concatenating polar codes with simple cyclic redundancy check (CRC), it was observed that the performance of the SCL decoder can be significantly improved \cite{Niu2012CRC}. With a large list size, the decoding performance provided by the CRC-aided SCL (CA-SCL) decoder can approach that of the ML decoder, which makes polar codes competitive with the other state-of-the-art channel codes, such as low-density parity-check (LDPC) and turbo codes \cite{Balatsoukas2017Comparison}.

The second drawback stems from the sequential nature of SC decoding. To tackle this issue, many parallel decoding schemes have been developed for polar codes \cite{Sarkis2013Increasing,Alamdar2011Simplified,Sarkis2014Fast,Hanif2017Fast,Condo2018Generalized,Sarkis2016Fast,Hashemi2016Fast,Hashemi2017Fast,Ardakani2019Fast,Zheng2021Threshold}. The main idea behind these schemes is to decode at the intermediate levels of the SC or SCL decoding tree, instead of the leaf nodes. In particular, it was shown in \cite{Sarkis2013Increasing} that ML decoding of the intermediate nodes can be employed to parallelize the SC decoding. However, this scheme is only suitable for nodes with short lengths. \textcolor{black}{On the contrary, the work \cite{Alamdar2011Simplified} proposed a low-complexity simplified SC (SSC) decoding algorithm for two certain frozen bit patterns, i.e., rate-zero (Rate-0) and rate-one (Rate-1) nodes that contain no information and frozen bit, respectively. In particular, the ML codeword of a Rate-0 node is always an all-zero vector, while the ML codeword of a Rate-1 node is the hard-decision output of the log-likelihood ratio (LLR) vector at the node root.} Likewise, single-parity-check (SPC) and repetition (REP) nodes along with their fast decoding techniques were proposed in \cite{Sarkis2014Fast}, which is known as the fast SSC (FSSC) decoding algorithm. Furthermore, the work \cite{Hanif2017Fast} advanced the studies in \cite{Alamdar2011Simplified} and \cite{Sarkis2014Fast} by proposing fast decoders for five new types of special nodes, namely the Type I-V nodes, which further reduces the SC decoding latency. In \cite{Condo2018Generalized}, a generalized REP (G-REP) node and a generalized parity-check (G-PC) node were proposed to reduce the latency of SC decoding even further. Moreover, the identification and utilization of the aforementioned special nodes were also extended to SCL decoding \cite{Sarkis2016Fast,Hashemi2016Fast,Hashemi2017Fast,Ardakani2019Fast}. However, all these works require the design of a separate decoder for each class of nodes, which inevitably increases the implementation complexity. In \cite{Zheng2021Threshold}, a class of sequence Rate-0/REP (SR0/REP) nodes was proposed which provides a unified description of most of the existing special nodes. With the introduction of SR0/REP nodes, a generalized fast SC decoding algorithm was proposed to achieve a higher degree of parallelism without degrading the error-correction performance.

In this paper, a new fast SC decoding algorithm is proposed to further reduce the decoding latency of polar codes. First, a new class of multi-node information and frozen bit patterns is introduced, which is composed of a sequence of Rate-1 or SPC (SR1/SPC) nodes, and thus provides a unified description of a wide variety of existing special nodes. The proposed SR1/SPC node is typically found at higher levels of the decoding tree, thus a higher degree of parallelism can be exploited as compared to the existing special nodes. Then, we investigate the impact of the frozen bits contained in the proposed SR1/SPC node, which leads to two types of parity constraints that are imposed on the decoded codeword. Furthermore, a simple and efficient decoding algorithm is proposed to decode the proposed SR1/SPC node, which can achieve near-ML performance with higher degree of parallelism. By combining the proposed SR1/SPC nodes with other types of special nodes, such as the SR0/REP nodes \cite{Zheng2021Threshold}, we show that the overall decoding latency can be reduced by 43.8\% and 62.9\% as compared to the state-of-the-art SC decoder and the renowned FSSC decoder \cite{Sarkis2014Fast}, respectively. 

The remainder of this paper is organized as follows. Section~\ref{section:2} reviews the backgrounds on polar codes, SC decoding, and fast SC decoding techniques. In Section~\ref{section:3}, we introduce the SR1/SPC node, and then analyse the induced parity constraints. Fast SC decoding of the proposed SR1/SPC node is presented in Section~\ref{section:4}. Section~\ref{section:5} provides simulation results to evaluate the decoding latency, complexity and performance. Finally, conclusions are drawn in Section~\ref{section:6}.

\emph{Notations}: Scalars, vectors, and matrices are respectively denoted by lower case, boldface lower case, and boldface upper case letters. For an arbitrary vector $\bm{a}$, {\color{black}$a[i:k:j]$ represents a subvector of $\bm{a}$ which is constructed by $(a[i], a[i+k], \ldots, a[i+m k])$}, where $k$ is the step size and $m = \lfloor (j-i)/k \rfloor$. If $k=1$, $a[i:k:j]$ is simply written as $a[i:j]$. $\operatorname{sgn}(a)$ denotes the sign of a scalar $a$ and $\operatorname{min}(\bm{a})$ returns the minimum element in a vector $\bm{a}$. Besides, $\bm{I}_N$ and $\bm{0}_N$ denote the $N \times N$ identity matrix and the $N \times N$ all-zero matrix, respectively. For a matrix $\bm{A}$, $(\bm{A})_{i}$ represents its $i$-th column vector. In addition, \textcolor{black}{$\{\cdot\}$ denotes a set}, $\lceil \cdot \rceil$ and $\lfloor \cdot \rfloor$ denote the round-up and round-down operations, respectively. $\oplus$ denotes the bitwise XOR operation and $\otimes$ denotes the Kronecker product. $\bm{A}^{\otimes n}$ denotes the $n$-th Kronecker power of $\bm{A}$.

\vspace{-0.5em}
\section{Preliminaries} \label{section:2}
\vspace{-0.5em}
\subsection{Polar Codes}
A polar code with code length $N=2^{n}$ and information length $K$ can be represented as $\mathcal{P}(N, K)$, which has a code rate of $R=K/N$. The transmitted polar codeword $\bm{x} = (x[1], x[2], \ldots, x[N])$ is obtained by $\bm{x} = \bm{u} \bm{G}_N$, where $\bm{u} = (u[1], u[2], \ldots, u[N])$ is the message vector and $\bm{G}_N = \bm{F}^{\otimes n}$ is the generator matrix with $\bm{F} = \begin{bmatrix}\begin{smallmatrix}	1 & 0 \\ 1 & 1 \end{smallmatrix}\end{bmatrix}$ being the base polarizing matrix. The message vector $\bm{u}$ is constructed by choosing $K$ bit-channels with high reliability to transmit information bits, while the remaining $N-K$ bits are frozen to some fixed values (usually set to 0). With the help of an indicator vector $\bm{c} = (c[1], c[2], \ldots, c[N])$, we are able to distinguish information and frozen bit-channels according to
\begin{equation}
	\begin{aligned}
		c[k] = \begin{cases} 1, & \textrm{if}\ k \in \mathcal{A} \\ 0, & \textrm{if}\ k \in \mathcal{A}^c\end{cases},
	\end{aligned}
	\label{eqn:1}
\end{equation}
where $\mathcal{A}$ and $\mathcal{A}^c$ denote the sets of information and frozen bits indices, respectively, which are both known to the encoder and decoder. The codeword vector $\bm{x}$ is then modulated and transmitted over the channel. Throughout this paper, we consider binary phase shift keying (BPSK) modulation and additive white Gaussian noise (AWGN) channel. 

\vspace{-0.5em}
\subsection{SC Decoding}
SC decoding originated in \cite{Arikan2009Channel} is a greedy search algorithm for decoding polar codes, and only the best decoding result is reserved in each decoding step. The decoding procedure of SC can be interpreted as a binary tree search process that starts from the root node to the leaf node and from the left branch to the right. At the $p$-th ($0 \leq p \leq n$) level of the decoding tree, each parent node, referred as $\mathcal{N}^{i}_{p}$, has a left child node $\mathcal{N}^{2i-1}_{p-1}$ and a right child node $\mathcal{N}^{2i}_{p-1}$, where $1 \leq i \leq 2^{n-p}$. There are two types of messages, i.e., the soft LLRs $\alpha^{i}_{p}[1:2^p]$ that are propagated from the parent node to their child nodes, and the hard codeword $\beta^{i}_{p}[1:2^p]$ that is propagated from the child nodes to their parent node in return. The $2^{p-1}$ LLRs of the left child node $\alpha^{2i-1}_{p-1}[1:2^{p-1}]$ and those of the right child node $\alpha^{2i}_{p-1}[1:2^{p-1}]$ can be respectively obtained by
\begin{equation}\color{black}
\begin{aligned}
\alpha^{2i-1}_{p-1}[k] =& \operatorname{sgn}(\alpha^{i}_{p}[k]) \operatorname{sgn}(\alpha^{i}_{p}[k+2^{p-1}])  \operatorname{min}(|\alpha^{i}_{p}[k]|, |\alpha^{i}_{p}[k+2^{p-1}]|), 
\end{aligned}
\label{eqn:2}
\end{equation}
\vspace{-2em}
\begin{equation}
\begin{aligned}
\alpha^{2i}_{p-1}[k] = \alpha^{i}_{p}[k+2^{p-1}] + (1-2\beta^{2i-1}_{p-1}[k]) \alpha^{i}_{p}[k].
\end{aligned}
\label{eqn:3}
\end{equation}
Besides, the hard codeword of $\mathcal{N}^{i}_{p}$, i.e., $\beta^{i}_{p}[1:2^p]$, is calculated based on the hard codewords of $\mathcal{N}^{2i-1}_{p-1}$ and $\mathcal{N}^{2i}_{p-1}$ as follows:
\begin{equation}
\begin{aligned}
\beta^{i}_{p}[k] = \begin{cases}\beta^{2i-1}_{p-1}[k] \oplus \beta^{2i}_{p-1}[k], & \textrm{if}\ 1 \leq k \leq 2^{p-1} \\
\beta^{2i}_{p-1}[k], & \textrm{otherwise}\end{cases},
\end{aligned}
\label{eqn:4}
\end{equation}
at the leaf level $p=0$, bit $u[k]$ can be estimated by performing hard decision on $\alpha^{k}_{0}[1]$ according to
\begin{equation}
\begin{aligned}
\hat{u}[k] = \beta^{k}_{0}[1] = \operatorname{HD}(\alpha^{k}_{0}[1]) =  \begin{cases}\frac{1-\textrm{sgn}(\alpha^{k}_{0}[1])}{2}, & \textrm{if}\ k \in \mathcal{A} \\ 0, & \textrm{if}\ k \in \mathcal{A}^c\end{cases},
\end{aligned}
\label{eqn:5}
\end{equation}
where $\hat{u}[k]$ is the estimate of $u[k]$ and $\operatorname{HD}(\cdot)$ is the hard decision function.

\vspace{-0.5em}
\subsection{Fast SC Decoding}
The sequential nature of SC decoding, i.e., each bit estimate depends on all previous ones, results in high decoding latency. \textcolor{black}{The decoding speed can be accelerated by implementing decoders that can output multiple bits in parallel without traversing the bottom of decoding tree.} Based on this idea, the work \cite{Sarkis2013Increasing} presented a highly parallel ML decoder to estimate the codeword of node $\mathcal{N}^{i}_{p}$, i.e., 
\begin{equation}
	\begin{aligned}
		\beta^{i}_{p}[1:2^p] = \mathop{\arg\max}\limits_{\beta^{i}_{p}[1:2^p] \in \mathbb{B}^{i}_{p}} \sum_{k=1}^{2^p} (-1)^{\beta^{i}_{p}[k]} \alpha^{i}_{p}[k],
	\end{aligned}
	\label{eqn:6}
\end{equation}
where $\mathbb{B}^{i}_{p}$ is the set of all codewords associated with node $\mathcal{N}^{i}_{p}$. However, the complexity of solving \eqref{eqn:6} can be very high, especially for nodes with long lengths. Therefore, the work \cite{Sarkis2014Fast} proposed the FSSC decoding algorithm, based on the discovery that the calculation of \eqref{eqn:6} can be significantly simplified for some nodes with special information and frozen bit patterns. In particular, four types of special nodes, i.e., Rate-0, Rate-1, REP and SPC, are considered in the FSSC decoder, and their structures are described as follows: 
\begin{itemize}
	\item Rate-0: all bits are frozen bits, $\bm{c} = \{0, 0, \ldots, 0\}$.
	\item Rate-1: all bits are information bits, $\bm{c} = \{1, 1, \ldots, 1\}$.
	\item REP: all bits are frozen bits except the rightmost one, $\bm{c} = \{0, \ldots, 0, 1\}$.
	\item SPC: all bits are information bits except the leftmost one, $\bm{c} = \{0, 1, \ldots, 1\}$.
\end{itemize}

Moreover, the work \cite{Hanif2017Fast} further improved the FSSC decoding speed by investigating five additional special nodes along with their efficient SC decoders. The five additional special nodes are summarized as follows:
\begin{itemize}
	\item Type I: all bits are frozen bits except the rightmost two, $\bm{c} = \{0, \ldots, 0, 1, 1\}$.
	\item Type II: all bits are frozen bits except the rightmost three, $\bm{c} = \{0, \ldots, 0, 1, 1, 1\}$.
	\item Type III: all bits are information bits except the leftmost two, $\bm{c} = \{0, 0, 1, \ldots, 1\}$.
	\item Type IV: all bits are information bits except the leftmost three, $\bm{c} = \{0, 0, 0, 1, \ldots, 1\}$.
	\item Type V: all bits are frozen bits except the rightmost three and the fifth-to-last, $\bm{c} = \{0, \ldots, $ $0, 1, 0, 1, 1, 1\}$.
\end{itemize}

In \cite{Condo2018Generalized}, two types of generalized fast decoding techniques for G-REP and G-PC nodes were introduced. G-REP is a node at level $p$ with all its descendants being Rate-0 nodes, except the rightmost one at a certain level $q < p$, which is a generic node of rate $C$ ($0 \leq C < 1$). Similarly, G-PC is a node at level $p$ having all its descendants as Rate-1 nodes except the leftmost one at a certain level $q < p$, which is a Rate-0 node. 

Recently, the work \cite{Zheng2021Threshold} proposed a new class of multi-node information and frozen bit patterns, namely SR0/REP node, which includes most of the existing special nodes as special cases. For an SR0/REP node at level $p$, all its descendants are Rate-0 or REP nodes except the rightmost one at level $q$, which is a generic source node. The general structure of the SR0/REP node is depicted in Fig.~\ref{fig:1}, where $\mathcal{N}^{R_q}_{q}$ denotes the source node at level $q$ with $R_q = 2^{p-q}i$. Note that an SR0/REP node will reduce to a G-REP node if all its descendants except the source node are Rate-0 nodes.

\begin{figure}[t]
	\centering
		\setlength{\abovecaptionskip}{-0.1cm} 
			\setlength{\belowcaptionskip}{-0.1cm} 
	\includegraphics[width=0.45\textwidth]{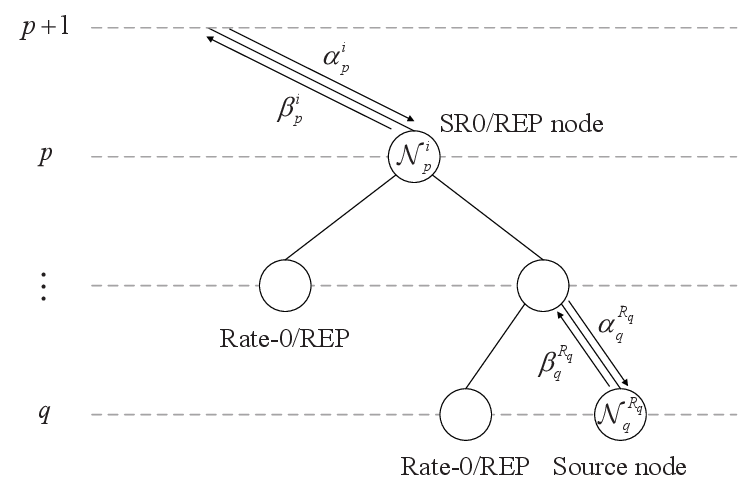}
	\caption{General structure of the SR0/REP node.}
	\label{fig:1}
	\vspace{-0.5cm}
\end{figure}

An SR0/REP node can be decoded based on the codeword of its source node. Let $\bm{s}_l = \{s_l[1], \ldots, s_l[2^{p-q}]\}$ denote the repetition sequence given by
\begin{equation}
	\begin{aligned}
		\bm{s}_l = (\eta_{q},1) \otimes (\eta_{q+1},1) \otimes \cdots \otimes (\eta_{p-1},1),\qquad l \in \{1,2,\ldots,|\mathbb{S}|\},
	\end{aligned}
	\label{eqn:7}
\end{equation}
where $\mathbb{S}$ denotes the set containing all possible $\bm{s}_l$, and $\eta_{r}$ denotes the rightmost bit value of the descendant Rate-0/REP node $\mathcal{N}^{R_r-1}_{r}$ at level $r$, i.e., 
\begin{equation}
	\begin{aligned}
		\eta_{k} = \begin{cases}
			1, & \textrm{if}\ \mathcal{N}^{R_r-1}_{r}\ \textrm{is a Rate-0 node} \\
			1-2\beta^{R_r-1}_{r}[2^r], & \textrm{if $\mathcal{N}^{R_r-1}_{r}$ is an REP node}
		\end{cases}.
	\end{aligned}
	\label{eqn:8}
\end{equation}
Note that $\bm{s}_l$ can be pre-determined based on the SR0/REP node structure before decoding. Let $\alpha^{R_q}_{q_l}[1:2^p]$ represent the LLR vector of the source node associated with $\bm{s}_l$, then in order to decode an SR0/REP node, the LLRs of the source node are first calculated as
\begin{equation}
	\begin{aligned}
		\alpha^{R_q}_{q_l}[k] = \sum\limits_{m=1}^{2^{p-q}} \alpha^{i}_{p}[(m-1)2^q+k] s_l[m], \qquad k \in \{1,2,\ldots,2^q\},
	\end{aligned}
	\label{eqn:9}
\end{equation}
Then, the optimal decoding path index $\hat{l}$ can be selected according to
\begin{equation}
	\begin{aligned}
		\hat{l} = \mathop{\arg\max}\limits_{l \in \{1,\ldots,|\mathbb{S}|\}}{\sum\limits_{k=1}^{2^q}|\alpha^{R_q}_{q_l}[k]|}.
	\end{aligned}
	\label{eqn:10}
\end{equation}
Subsequently, the source node is decoded to obtain $\beta^{R_q}_{q}[k]$ using $\alpha^{R_q}_{q_{\hat{l}}}[1:2^p]$. In particular, fast decoding techniques can be utilized if the source node has a special structure, otherwise the plain SC decoding is used. Finally, the decoding result of the SR0/REP node is obtained by
\begin{equation}
	\begin{aligned}
		\beta^{i}_{p}[k:2^q:2^p] = \beta^{R_q}_{q}[k] \oplus \bm{s}_{\hat{l}}, \qquad k \in \{1,2,\ldots,2^q\}.
	\end{aligned}
	\label{eqn:11}
\end{equation}

Compared with the plain SC decoding, all of the aforementioned fast SC decoding algorithms can preserve the decoding performance with reduced latency by exploiting the special structure of the nodes at intermediate levels of the decoding tree. As a result, if the level of the special node to be decoded is higher, the decoding speed would be faster, but the decoding algorithm might be more difficult to be conducted since the associated bit pattern is more complicated. 

\vspace{-0.5em}
\section{Sequence Rate-1 or SPC nodes}
\label{section:3}
In this section, we introduce a new special node, known as the SR1/SPC node, which is characterized by a source node combined with a sequence of Rate-1 or SPC nodes. Then, we characterize the parity constraints induced by the proposed special node that should be satisfied in order to facilitate fast decoding without performance degradation. 

\vspace{-0.5em}
\subsection{Node Structure}
Similar to the SR0/REP node \cite{Zheng2021Threshold}, the proposed SR1/SPC node is defined as any node at level $p$ whose \textcolor{black}{right descendants at level $q \leq r < p$} are all Rate-1 or SPC nodes, except the leftmost one at a certain level $q$, \textcolor{black}{where $0 \leq q < p$}. The leftmost node is referred as the source node, which is a generic node of Rate-C. For illustration purpose, the structure of the SR1/SPC node is depicted in Fig.~\ref{fig:2}.

\begin{figure}[!t]
	\centering
	\setlength{\abovecaptionskip}{-0.1cm} 
	\setlength{\belowcaptionskip}{-0.1cm} 
	\includegraphics[width=0.45\textwidth]{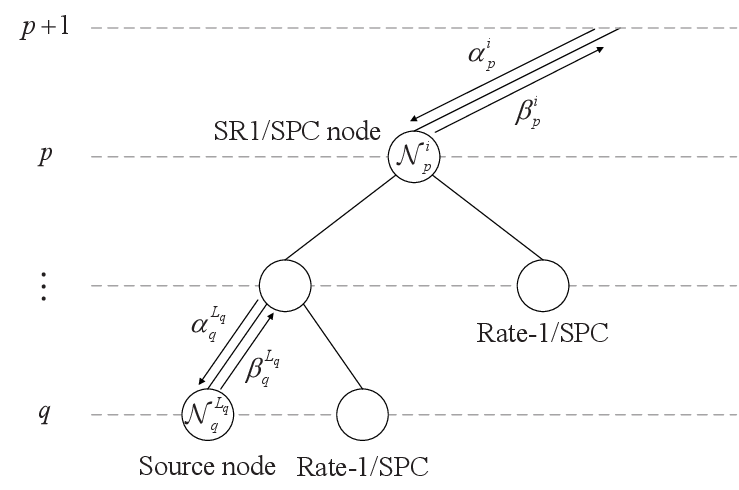}
	\caption{General structure of an SR1/SPC node.}
	\label{fig:2}
	\vspace{-0.5cm}
\end{figure}

Given an SR1/SPC node, \textcolor{black}{its node structure, denoted by $\operatorname{NS}(p,q,\mathcal{L})$}, is characterized by three key parameters, where $p$ is the root node level, $q$ is the source node level, and $\mathcal{L}$ represents a level index set which consists of the level indices of each descendant SPC node, i.e., $\mathcal{L} = \{r | \mathcal{N}^{L_r+1}_{r}\ \textrm{is an SPC node}, q \leq r < p\}$ with $L_r = 2^{p-r}(i-1)+1$. Note that the decoding latencies of the existing fast SC decoding algorithms increase linearly with $p-q$, therefore we define $d=p-q$ as the node depth to reflect the parallelism of the SR1/SPC node. Besides, without loss of generality, we assume that the level indices are sorted in ascending order. For instance, a $\mathcal{P}(32,27)$ polar code with 
\begin{equation*}
	\bm{c} = \{\overbrace{0,0,0,1}^{\textrm{REP}},\overbrace{0,1,1,1}^{\textrm{SPC}},\overbrace{0,1,\ldots,1}^{\textrm{SPC}},\overbrace{1,\ldots,1}^{\textrm{Rate-1}}\}
\end{equation*}
can be represented as an SR1/SPC node with $\operatorname{NS}(5,2,\{2,3\})$. 

\begin{table*}[tb] \scriptsize\color{black}
	\caption{Special Cases of SR1/SPC Node}
	\centering
	\textcolor{black}{
	\begin{tabular}{ccc}\hline
		\multirow{2}{*}{Node Type} & \multicolumn{2}{c}{Node Structure} \\ \cline{2-3} & Parameters & Source Node \\ \hline
		P-01 \cite{Ercan2019Operation} & $\operatorname{NS}(p,p-1,{\emptyset})$ & Rate-0 \\
		P-0SPC \cite{Ercan2019Operation} & $\operatorname{NS}(p,p-1,\{p-1\})$ & Rate-0 \\
		P-R1 \cite{Ercan2019Operation} & $\operatorname{NS}(p,p-1,{\emptyset})$ & Rate-C \\
		P-RSPC \cite{Ercan2019Operation} & $\operatorname{NS}(p,p-1,\{p-1\})$ & Rate-C \\
		Type III \cite{Hanif2017Fast} & $\operatorname{NS}(p,1,{\emptyset})$ & Rate-0 \\
		Type IV \cite{Hanif2017Fast} & $\operatorname{NS}(p,2,{\emptyset})$ & REP \\
		G-PC \cite{Condo2018Generalized} & $\operatorname{NS}(p,q,{\emptyset})$ & Rate-0 \\
		RG-PC \cite{Condo2018Generalized} & $\operatorname{NS}(p,q,\mathcal{L}), \mathcal{L} \not= \emptyset$ & Rate-0 \\
		EG-PC \cite{Zheng2021Threshold} & $\operatorname{NS}(p,q,{\emptyset})$ & Rate-0/REP \\ 
		SR1 & $\operatorname{NS}(p,q,{\emptyset})$ & Rate-C \\
		SSPC & $\operatorname{NS}(p,q,\{q,q+1,\ldots,p-1\})$ & Rate-C \\ \hline
	\end{tabular}}
	\label{tab1}
	\vspace{-0.2cm}
\end{table*}

{\color{black}Since the source node can be any generic node, the proposed SR1/SPC node can envelop various kinds of special nodes. As shown in Table~\ref{tab1}, most of the existing special nodes can be viewed as special cases of the SR1/SPC node. However, the existing special nodes either have smaller node depth (parallelism), or are rarely distributed (see Table~\ref{tab2}), which both lead to limited decoding latency reduction. Note that the case of $\mathcal{L} = \{q,q+1,\ldots,p-1\}$ or $\mathcal{L} = \emptyset$, an SR1/SPC node will reduce to a sequence Rate-1 (SR1) or a sequence SPC (SSPC) node, respectively. Besides, it is worth mentioning that exploring other types of descendant nodes may lead to the discovery of more special node types that can be decoded efficiently, however further exploration is left for future work.}

\vspace{-0.5em}
\subsection{Induced Parity Constraints}
{\color{black}
As pointed out in \cite{Condo2018Generalized}, each frozen bit at the leaf level of the decoding tree will impose a parity constraint on the possible codewords at the intermediate levels. This means that the codeword at the node root need to satisfy the parity constraints to ensure its validity, otherwise there would be decoding performance loss. For instance, the decoding of the RG-PC node introduced in \cite{Condo2018Generalized} ignored the parity constraints brought by the frozen bits in the descendant nodes, which leads to significant performance degradation at high SNR scenarios. In the following, we analyse the parity constraints imposed on the root of the SR1/SPC node.}

For an arbitrary SR1/SPC node, the parity constraints induced by the frozen bits in the source node can be obtained as follows:
\begin{theorem} \label{theorem:1}
	For an SR1/SPC node $\mathcal{N}^{i}_{p}$, the source node at level $q$ will impose the following $2^q$ parallel parity constraints (P-PC) on the root node at level $p$:
	\begin{equation}
		\begin{aligned}
			\bigoplus\limits_{j=1}^{2^d} \beta^{i}_{p}[(j-1)2^q+k] = \beta^{L_q}_{q}[k], \qquad k \in \{1,2,\ldots,2^q\}.
		\end{aligned}
		\label{eqn:12}
	\end{equation}
\end{theorem}
\begin{proof}
Please see Appendix~\ref{appendex:1}.
\end{proof}

For an SR1 node, since there is no other frozen bit except for those in the source node, we do not need to consider other parity constraints except the P-PCs in \eqref{eqn:12}. However, for a general SR1/SPC node, each of its descendant SPC node will induce an additional parity constraint on its root node, which is shown as follows:

\begin{theorem}	\label{theorem:2}
	For an SR1/SPC node $\mathcal{N}^{i}_{p}$, the descendant SPC node at level $r$, i.e., $N_r^{L_r+1}$, will impose a segmental parity constraint (S-PC) on the root node at level $p$, which is given by 
	\begin{equation}
		\begin{aligned}
			\bigoplus\limits_{j=1}^{2^{p-r-1}} \bigoplus\limits_{k=1}^{2^r} \beta^{i}_{p}[(2j-1)2^r+k] = 0. 
		\end{aligned}
		\label{eqn:13}
	\end{equation}
\end{theorem}
\begin{proof}
Please see Appendix~~\ref{appendex:2}.
\end{proof}

\begin{figure*}[!t] 
	\centering
	\setlength{\abovecaptionskip}{-0.1cm} 
	\setlength{\belowcaptionskip}{-0.1cm} 
	\includegraphics[width=0.95\textwidth]{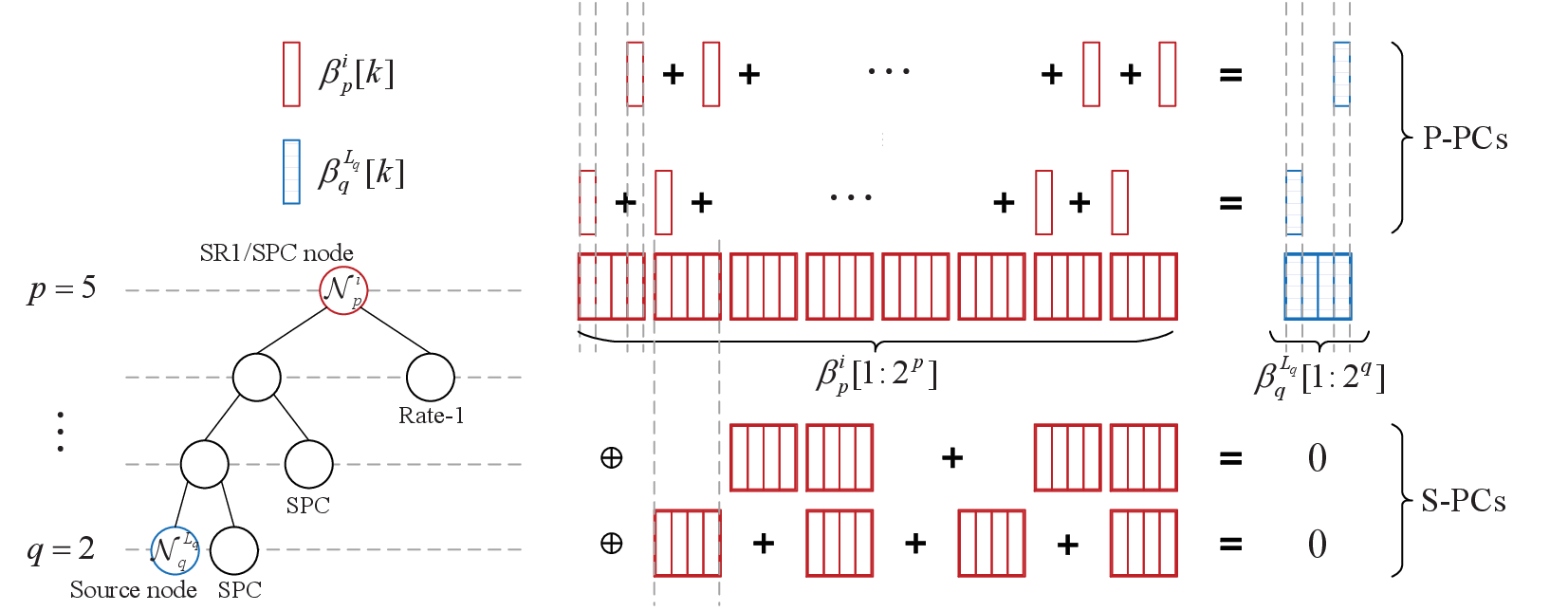}
	\caption{Parity constraints of an SR1/SPC node $\operatorname{NS}(5,2,\{2,3\})$.} \label{fig:3}
	\vspace{-0.5cm}
\end{figure*}

According to Theorems~\ref{theorem:1} and \ref{theorem:2}, in order to decode an SR1/SPC node, the aforementioned two types of parity constraints, i.e., P-PCs and S-PCs, derived from the source and descendant SPC nodes, should both be taken into consideration. For illustration purpose, Fig.~\ref{fig:3} depicts all the P-PCs and S-PCs of an SR1/SPC node denoted by $\operatorname{NS}(5,2,\{2,3\})$. It can be observed that the source node with length 4 imposes 4 P-PCs on the root node, while the descendant SPC nodes at levels 2 and 3 introduce two separate S-PCs. Besides, we note that the P-PCs introduce parity constraints on the root bits that are equally spaced in $\beta_p^i[1:2^p]$, whereas the bits involved in the S-PCs are located in different segments of the root node.
\vspace{-0.5em}
\textcolor{black}{
\begin{remark}
\emph{It is also worth mentioning that the induced parity constrains only depend on the SR1/SPC node structure and are irrespective of the employed decoding algorithms. As a result, when the SR1/SPC node is adopted in the case of other SC-based decoding (e.g., the renowned SCL and SC-Flip \cite{Afisiadis2014low} decoders), each path should also satisfy the P-PCs and S-PCs in order to ensure its validity. Such considerations may facilitate other fast SC-based decoding of the SR1/SPC node, which is left for future work.}
\end{remark}
}

\vspace{-0.5em}
\section{Fast Decoding of SR1/SPC Nodes}
\label{section:4}
{\color{black}
The results in Theorems~\ref{theorem:1} and \ref{theorem:2} indicate that the decoding output of SR1/SPC node should satisfy the P-PCs and S-PCs simultaneously. However, each S-PC is dependent on the P-PCs and other S-PCs due to the bits in common, which makes it difficult to develop an efficient decoding algorithm since flipping one bit may impact the parity checks of both P-PCs and S-PCs. To tackle this issue, we present an efficient decoding algorithm with two stages to progressively satisfy the P-PCs and S-PCs. In the first stage, the P-PCs are corrected by temporarily ignoring the S-PCs, and in the second stage, the S-PCs are corrected without violating the P-PCs, such that all the parity check constraints are satisfied simultaneously. In each stage, we apply bit-flipping operations on the hard-decision codeword to obtain candidate codewords. Based on \eqref{eqn:6}, a penalty metric is proposed to measure the reliability of each candidate codeword. Finally, we select the codeword which is penalised the least as the decoding output, and thus the near-ML decoding performance can be achieved. Note that existing highly-parallel ML decoding variants, such as the guessing random additive noise decoding algorithm \cite{Duffy2019Capacity} and the syndrome-check decoding algorithm \cite{Hashemi2021Successive} can also achieve near-ML performance, but are usually computational-unfriendly.

In the following, we will present the proposed decoding algorithm for the SR1/SPC nodes, which are divided into two stages, and then a simplification technique is further proposed to reduce the computational complexity. }

\vspace{-0.5em}
\subsection{Correcting the P-PCs}
First, we divide both the LLRs and codeword of the SR1/SPC node into $2^{q}$ parts such that
\begin{equation}
	\begin{aligned}
		\alpha^{\textrm{P-PC}}_{k}[1:2^d] = \alpha^{i}_{p}[k:2^{q}:2^p], \qquad k \in \{1,2,\ldots,2^{q}\}, \\
		\beta^{\textrm{P-PC}}_{k}[1:2^d] = \beta^{i}_{p}[k:2^{q}:2^p], \qquad k \in \{1,2,\ldots,2^{q}\},
	\end{aligned}
	\label{eqn:14}
\end{equation}
where $\alpha^{\textrm{P-PC}}_{k}[1:2^d]$ denotes the LLR subvector associated with the $k$-th P-PC, and $\beta^{\textrm{P-PC}}_{k}[1:2^d]$ represents the corresponding codeword of $\alpha^{\textrm{P-PC}}_{k}[1:2^d]$. Then, the LLR vector $\alpha^{L_{q}}_{q}[1:2^{q}]$ of the source node can be obtained by 
\begin{equation}
	\begin{aligned}
		\alpha^{L_{q}}_{q}[k] = \prod\limits_{j=1}^{2^d} \operatorname{sgn}(\alpha^{\textrm{P-PC}}_{k}[j]) \operatorname{min}(|\alpha^{\textrm{P-PC}}_{k}[1:2^d]|), \qquad  k \in \{1,2,\ldots,2^{q}\}.
	\end{aligned}
	\label{eqn:15}
\end{equation}

Subsequently, the source node is decoded either by plain SC decoding, or by fast decoding techniques if it has a special pattern, to obtain the codeword $\beta^{L_{q}}_{q}[1:2^{q}]$. According to Theorem~\ref{theorem:1}, each LLR sub-vector $\alpha^{\textrm{P-PC}}_{k}[1:2^d]$ can be interpreted as an SPC subcode and the corresponding codeword bits should satisfy the parity check $\beta^{L_{q}}_{q}[k]$.\footnote{\color{black}Slightly different from the definition of the conventional SPC node, the used `SPC' actually refers to a generalized SPC node. The conventional one needs to satisfy an even parity check, whereas the parity check of a generalized SPC node can be either even or odd, depending on $\beta^{L_{q}}_{q}[k]$.} In particular, parity check and bit-flipping will be performed by the Wagner decoder to decode these SPC subcodes \cite{Silverman1954Coding}. For the $k$-th SPC subcode, the parity check of the $k$-th P-PC, \textcolor{black}{denoted by $\gamma_{\textrm{P-PC}}[k]$}, and the least reliable position for bit-flipping, \textcolor{black}{denoted by $\rho[k]$}, are respectively determined by
\begin{equation}
	\begin{aligned}
		\gamma_{\textrm{P-PC}}[k] = \bigoplus\limits_{j=1}^{2^d} \operatorname{HD}(\alpha^{\textrm{P-PC}}_{k}[j]) \oplus \beta^{L_{q}}_{q}[k], 
	\end{aligned}
	\label{eqn:16}
\end{equation}
\begin{equation}\color{black}
	\begin{aligned}
		\rho[k] = \mathop{\arg\min}\limits_{j \in \{1,\ldots,2^d\}}{|\alpha^{\textrm{P-PC}}_{k}[j]|}.
	\end{aligned}
	\label{eqn:17}
\end{equation}
\vspace{-0.5em}
Finally, the codeword of the SR1/SPC node with all P-PCs satisfied is obtained by applying bit-flipping operations on the hard-decision codeword, which is shown as follows:
\begin{equation}
	\begin{aligned}
		\beta^{\textrm{P-PC}}_{k}[j] = \begin{cases}
			\operatorname{HD}(\alpha^{\textrm{P-PC}}_{k}[j]) \oplus \gamma_{\textrm{P-PC}}[k], & \textcolor{black}{\textrm{if}\ j = \rho[k]} \\ 
			\operatorname{HD}(\alpha^{\textrm{P-PC}}_{k}[j]), & \textrm{otherwise}
		\end{cases}.
	\end{aligned}
	\label{eqn:18}
\end{equation}

\begin{algorithm}[!t]
	\label{alg1}
	\caption{Correcting the P-PCs}
	\LinesNumbered
	\KwIn{LLR vector $\alpha^{i}_{p}[1:2^p]$ of the SR1/SPC node $\mathcal{N}^{i}_{p}$, node parameters $\operatorname{NS}(p,q,\mathcal{L})$}
	\KwOut{Codeword vector $\beta^{i}_{p}[1:2^p]$ of $\mathcal{N}^{i}_{p}$}
	
	\If{$\mathcal{N}^{L_{q}}_{q}$ is a Rate-0 node}{
		$\beta^{L_{q}}_{q}[1:2^{q}] = \mathbf{0}$\;}
	\Else{
		Calculate $\alpha^{L_{q}}_{q}[1:2^{q}]$ using \eqref{eqn:15}\;
		Obtain $\beta^{L_{q}}_{q}[1:2^{q}]$ by decoding $\mathcal{N}^{L_{q}}_{q}$ using $\alpha^{L_{q}}_{q}[1:2^q]$\;
	}
	Obtain $\beta^{i}_{p}[1:2^p]$ according to \eqref{eqn:18}\;

	\Return $\beta^{i}_{p}[1:2^p]$
\end{algorithm}

The procedure of correcting P-PCs is summarized in Algorithm~\ref{alg1}. Note that the identification of the source node can be performed offline before decoding, according to \eqref{eqn:18}. In Algorithm~\ref{alg1}, the source node is decoded according to its node type (lines 1-5). In particular, in case it is an Rate-0 node, its codeword can be obtained directly (line 2) since there is no information bit to be estimated. Otherwise, LLR calculation is performed (line 4), followed by a specific decoding procedure to obtain the codeword of the source node (line 5). Since the source node might be composed of one or several special nodes, the fast decoding techniques in \cite{Alamdar2011Simplified,Sarkis2013Increasing,Sarkis2014Fast,Hanif2017Fast,Condo2018Generalized,Zheng2021Threshold}, can be applied. Finally, in line 6, $2^{q}$ parallel Wagner decoders are implemented to decode the root node according to \eqref{eqn:18}. Note that the SR1 node can be directly decoded by Algorithm~\ref{alg1} since it contains no S-PC.

\vspace{-0.5em}
\subsection{Correcting the S-PCs}
Similar to \eqref{eqn:14}, the LLRs and codeword of the SR1/SPC node are divided into $2^d$ parts such that
\begin{equation}
	\begin{aligned}
		\alpha^{\textrm{S-PC}}_{k}[1:2^q] = \alpha^{i}_{p}[(k-1)2^q+1:k2^q], &\qquad  k \in \{1,2,\ldots,2^d\} \\
		\beta^{\textrm{S-PC}}_{k}[1:2^q] = \beta^{i}_{p}[(k-1)2^q+1:k2^q], &\qquad  k \in \{1,2,\ldots,2^d\}, 
	\end{aligned}
	\label{eqn:19}
\end{equation}
where $k$ denotes the segment index, $\alpha^{\textrm{S-PC}}_{k}[1:2^q]$ denotes the LLR subvector associated with the root bits from $(k-1)2^q+1$ to $k2^q$ in the $k$-th segment, and $\beta^{\textrm{S-PC}}_{k}[1:2^q]$ represents the corresponding codeword of $\alpha^{\textrm{S-PC}}_{k}[1:2^q]$. For the SR1/SPC node of depth $d$, let $\gamma_{\textrm{S-PC}}[1:d]$ denote the S-PC parity check vector, which can be calculated based on Theorem~\ref{theorem:2} as follows:
\begin{equation}
	\begin{aligned}
		\gamma_{\textrm{S-PC}}[t] = \begin{cases}
			\bigoplus\limits_{j=1}^{2^{d-t}} \bigoplus\limits_{k=1}^{2^{q+t-1}} \beta^{i}_{p}[(2j-1)2^{q+t-1}+k] = \bigoplus\limits_{j=1}^{2^{d-t}}\bigoplus\limits_{k=(2j-1)2^{t-1}+1}^{j2^t}\bigoplus\limits_{m=1}^{2^q} \beta^{\textrm{S-PC}}_{k}[m] \\ \hspace{0.45\textwidth} \textrm{if}\  \mathcal{N}^{L_r+1}_{r}\ \textrm{is an SPC node}\\
			-1, \hspace{0.4\textwidth} \textrm{if}\  \mathcal{N}^{L_r+1}_{r}\ \textrm{is a Rate-1 node}
		\end{cases},
	\end{aligned}
	\label{eqn:20}
\end{equation}
where $t = r-q+1 \in \{1,2,\ldots,d\}$ and $-1$ indicates that the descendant node $\mathcal{N}^{L_r+1}_{r}$ has no S-PC as it is a Rate-1 node. Note that $\gamma_{\textrm{S-PC}}[t]$ is also able to indicate whether the $t$-th S-PC is satisfied or not, since a satisfied S-PC leads to an even check and an unsatisfied one leads to the opposite. Besides, for an arbitrary even numbers of segments, the $m$-th bits in these segments should be flipped to maintain the $m$-th P-PCs ($m \in \{1,2\ldots,2^q\}$). Thus, let $\mathcal{E}=\{k_1,k_2,m\}$ denote a flip coordinate set, which indicates a pair of bit positions to be flipped, i.e., $\beta^{\textrm{S-PC}}_{k_1}[m]$ and $\beta^{\textrm{S-PC}}_{k_2}[m]$, where $k_1,k_2 \in \{1,2,\ldots,2^d\}$ are segment indices and $k_1 < k_2$. For simplicity, we drop the bit index $m$ and $\mathcal{E}=\{k_1,k_2,m\}$ is simply written as $\mathcal{E}=\{k_1,k_2\}$ in the following. Based on $\mathcal{E}$, we are able to determine the feasible flip coordinates and flipping the corresponding bits can correct the S-PCs with odd checks without violating the P-PCs.

As mentioned above, a feasible flip coordinate should correct the odd S-PC checks while maintain the even ones. Therefore, for a feasible flip coordinate $\mathcal{E}=\{k_1,k_2\}$, only one of the two segment indices should be involved in the S-PCs with odd checks. Besides, $k_1$ and $k_2$ should both be involved or not involved by the remaining S-PCs such that these S-PC checks can be maintained. For simplicity, we say two segments are equivalent if they are both involved or not involved in the $t$-th S-PC. On the contrary, two segments are not equivalent if only one of them is involved in the $t$-th S-PC. From \eqref{eqn:20}, we can readily determine whether a segment is related to the $t$-th S-PC or not. Specifically, segments with indices in
\begin{equation}
	\begin{aligned}
		\mathcal{I}_{t} = \{&2^{t-1}+1,2^{t-1}+2,\ldots,2^t,\ldots, 
		(2j-1)2^{t-1}+1,(2j-1)2^{t-1}+2,\ldots,j2^t,\ldots, \\
		&2^d-2^{t-1}+1,2^d-2^{t-1}+2,\ldots,2^d\},
	\end{aligned}
	\label{eqn:21}
\end{equation}
are involved in the $t$-th S-PC, whereas those with indices in
\begin{equation}
	\begin{aligned}
		\mathcal{I}_{t}^c = \{&1,2,\ldots,2^{t-1},\ldots, 
		(j\!-\!1)2^{t-1}\!+\!1,(j\!-\!1)2^{t-1}\!+\!2,\ldots,(2j\!-\!1)2^{t-1},\ldots, \\
		&2^d-2^t+1,2^d-2^t+2,\ldots,2^d-2^{t-1}\},
	\end{aligned}
	\label{eqn:22}
\end{equation}
are not involved, where $j \in \{1,2,\ldots,2^{d-t}\}$. By identifying the segments that impact each S-PC, we can obtain the following theorem to determine whether two segments are equivalent or not.
\begin{theorem}
	For the $t$-th S-PC, segments $k+\mu2^t$ with $\mu \in \{\lceil -k/2^t \rceil,\ldots,\lfloor (2^d-k)/2^t \rfloor\} \setminus \{0\}$ are equivalent to segment $k$. Besides, segments $k+(2\nu-1)2^{t-1}$ with $\nu \in \{\lceil -k/2^t+1/2 \rceil,\ldots,\lfloor (2^d-k)/2^t+1/2 \rfloor\}$ are not equivalent to segment $k$.
	\label{theorem:3}
\end{theorem}
\vspace{-0.4cm}
\begin{proof}
Please refer to Appendix~\ref{appendex:3}.
\end{proof}

Theorem~\ref{theorem:3} implies that applying bit-flipping on the coordinate $\{k,k+\mu2^t\}$ is able to maintain the $t$-th S-PC, while performing bit-flipping on $\{k,k+(2\nu-1)2^{t-1}\}$ changes the $t$-th S-PC. For the special case that $\gamma_{\textrm{S-PC}}[1:d] = (0,\ldots,0,1)$, the feasible flip coordinates can be determined by resorting to the following lemma. 
\begin{lemma} \label{lemma:1}
	For SR1/SPC nodes, the feasible flip coordinates associated with $\gamma_{\textrm{S-PC}}[1:d] = (0,\ldots,0,1)$ are
	\begin{equation}
		\begin{aligned}
			\mathcal{E} = \{k,k+2^{d-1}\}, \qquad k \in \{1,2,\ldots,2^{d-1}\}.
		\end{aligned}
		\label{eqn:23}
	\end{equation}
\end{lemma}
\begin{proof}
	A feasible flip coordinate must contain the indices of two equivalent segments in order to maintain the S-PCs with even checks from $t=1$ to $t=d-1$. Besides, if the parity check of the $d$-th S-PC is odd, then it must be corrected by flipping two inequivalent segments. Starting from $t=1$, the indices of involved and uninvolved segments can be respectively obtained as $\mathcal{I}_{1}=\{2,4,6,\ldots,2^d\}$ and $\mathcal{I}_{1}^c=\{1,3,5,\ldots,2^d-1\}$ according to \eqref{eqn:21} and \eqref{eqn:22}. Thus, we can infer that for any feasible flip coordinate, the interval between two segment indices must be even, thus the set of feasible intervals can be determined by $\{2,4,6,\ldots,2^d\}$. Then, for the case of $t=2$, we can see that all the intervals with values $(2\nu-1)2^{t-1}$ are infeasible based on Theorem~\ref{theorem:3}, which reduces the set of feasible intervals to $\{4,8,12,\ldots,2^d\}$. Meanwhile, $\{4,8,12,\ldots,2^d\}$ is valid since two segments with intervals $\mu2^t$ are equivalent according to Theorem~\ref{theorem:3}. Furthermore, for the case of $t=3$, this set becomes $\{8,16,24,\ldots,2^d\}$ by resorting to Theorem~\ref{theorem:3} again. Repeat the above procedure multiple times until $t=d-1$, we can obtain that the only feasible interval is $2^{d-1}$. Finally, for the case of $t=d$, it can be readily proved that the segments with interval $2^{d-1}$ are inequivalent according to Theorem~\ref{theorem:3}. This thus completes the proof.
\end{proof}

Lemma~\ref{lemma:1} only provides the feasible flip coordinates for the special case of $\gamma_{\textrm{S-PC}}[1:d] = (0,\ldots,0,1)$. Next, we will construct the feasible flip coordinates for more general cases via induction. Given an depth-$(d-1)$ SR1/SPC node and a feasible flip coordinate $\mathcal{E}=\{k_1,k_2\}$ associated with $\gamma_{\textrm{S-PC}}[1:d-1]$, let us consider an extended SR1/SPC node of depth $d$, whose left child node is the original SR1/SPC node of depth $d-1$ and the right child node is a Rate-1 or an SPC node, which introduces an additional S-PC with parity check $\gamma_{\textrm{S-PC}}[d]$. Then, the feasible flip coordinates for the new SR1/SPC node of depth $d$ can be constructed based on the original one of depth $d-1$, which are shown as follows:
\begin{lemma}
	For SR1/SPC nodes, given a feasible flip coordinate $\mathcal{E}=\{k_1,k_2\}$ ($k_1,k_2 \in \{1,2,\ldots,$ $2^{d-1}\}$) associated with $\gamma_{\textrm{S-PC}}[1:d-1]$ and define four new flip coordinates $\mathcal{E}_1=\{k_1,k_2\}, \mathcal{E}_2=\{k_1+2^{d-1},k_2+2^{d-1}\}, \mathcal{E}_3=\{k_1+2^{d-1},k_2\}, \mathcal{E}_4=\{k_1,k_2+2^{d-1}\}$ based on $\mathcal{E}$, then the following flip coordinates associated with $\gamma_{\textrm{S-PC}}[1:d]$ are feasible for different cases of $\gamma_{\textrm{S-PC}}[d]$:
	\begin{equation}
		\begin{aligned}			
			\mathcal{E}_1, \mathcal{E}_2,\quad \text{if}\ \gamma_{\textrm{S-PC}}[d] = 0,
		\end{aligned}
		\label{eqn:24}
	\end{equation}
	\begin{equation}
		\begin{aligned}
			\mathcal{E}_3, \mathcal{E}_4,\quad \text{if}\ \gamma_{\textrm{S-PC}}[d] = 1,
		\end{aligned}
		\label{eqn:25}
	\end{equation}
	\begin{equation}
		\begin{aligned}
			\mathcal{E}_1, \mathcal{E}_2,
			\mathcal{E}_3, \mathcal{E}_4,\quad \text{if}\ \gamma_{\textrm{S-PC}}[d] = -1.
		\end{aligned}
		\label{eqn:26}
	\end{equation}
	\label{lemma:2}
	\vspace{-0.4cm}
\end{lemma}

\begin{proof}
	Firstly, for the case of $\gamma_{\textrm{S-PC}}[d] = 0$, the two segment indices in a feasible flip coordinate should be equivalent, while keeping the parity checks of the other S-PCs unchanged. It can be observed from \eqref{eqn:21} and \eqref{eqn:22} that $k_1$ and $k_2$ are both not involved in the $d$-th S-PC, while $k_1+2^{d-1}$ and $k_2+2^{d-1}$ are both involved, thus $\mathcal{E}_1=\{k_1,k_2\}$ and $\mathcal{E}_2=\{k_1+2^{d-1},k_2+2^{d-1}\}$ can maintain the $d$-th S-PC. Besides, according to Theorem~\ref{theorem:3}, segments $k$ and $k+2^{d-1}$ are equivalent for the $t$-th S-PC with $t \leq d-1$, which means that $\mathcal{E}_2=\{k_1+2^{d-1},k_2+2^{d-1}\}$ will not violate the other S-PCs. 
	
	Next, if $\gamma_{\textrm{S-PC}}[d] = 1$, the two segment indices in a feasible flip coordinate should be inequivalent, while keeping the parity checks of the other S-PCs unchanged. It can be observed from \eqref{eqn:21} and \eqref{eqn:22} that $k_1$ and $k_2$ are both not involved in the $d$-th S-PC, while $k_1+2^{d-1}$ and $k_2+2^{d-1}$ are both involved, thus $\mathcal{E}_3=\{k_1+2^{d-1},k_2\}$ and $\mathcal{E}_4=\{k_1,k_2+2^{d-1}\}$ can correct the parity check of the $d$-th S-PC. Besides, according to Theorem~\ref{theorem:3}, segments $k$ and $k+2^{d-1}$ are equivalent for the $t$-th S-PC with $t \leq d-1$, which means that both $\mathcal{E}_3=\{k_1+2^{d-1},k_2\}$ and $\mathcal{E}_4=\{k_1,k_2+2^{d-1}\}$ will not violate the other S-PCs.

	Finally, since the additional descendant node imposes no S-PC on the codeword of the SR1/SPC node for the case of $\gamma_{\textrm{S-PC}}[d] = -1$, a feasible flip coordinate associated with $\gamma_{\textrm{S-PC}}[1:d]$ only needs to keep all S-PC checks unchanged. As proved above, segments $k$ and $k+2^{d-1}$ are equivalent for the $t$-th S-PC with $t \leq d-1$, thus applying this equivalence to either one or both two of the elements in $\mathcal{E}=\{k_1,k_2\}$ results in four new feasible flip coordinates as shown in \eqref{eqn:26}, which completes the proof.
\end{proof}

\begin{figure*}[!t] 
	\setlength{\abovecaptionskip}{-0.1cm} 
	\centering
	\includegraphics[width=1\textwidth]{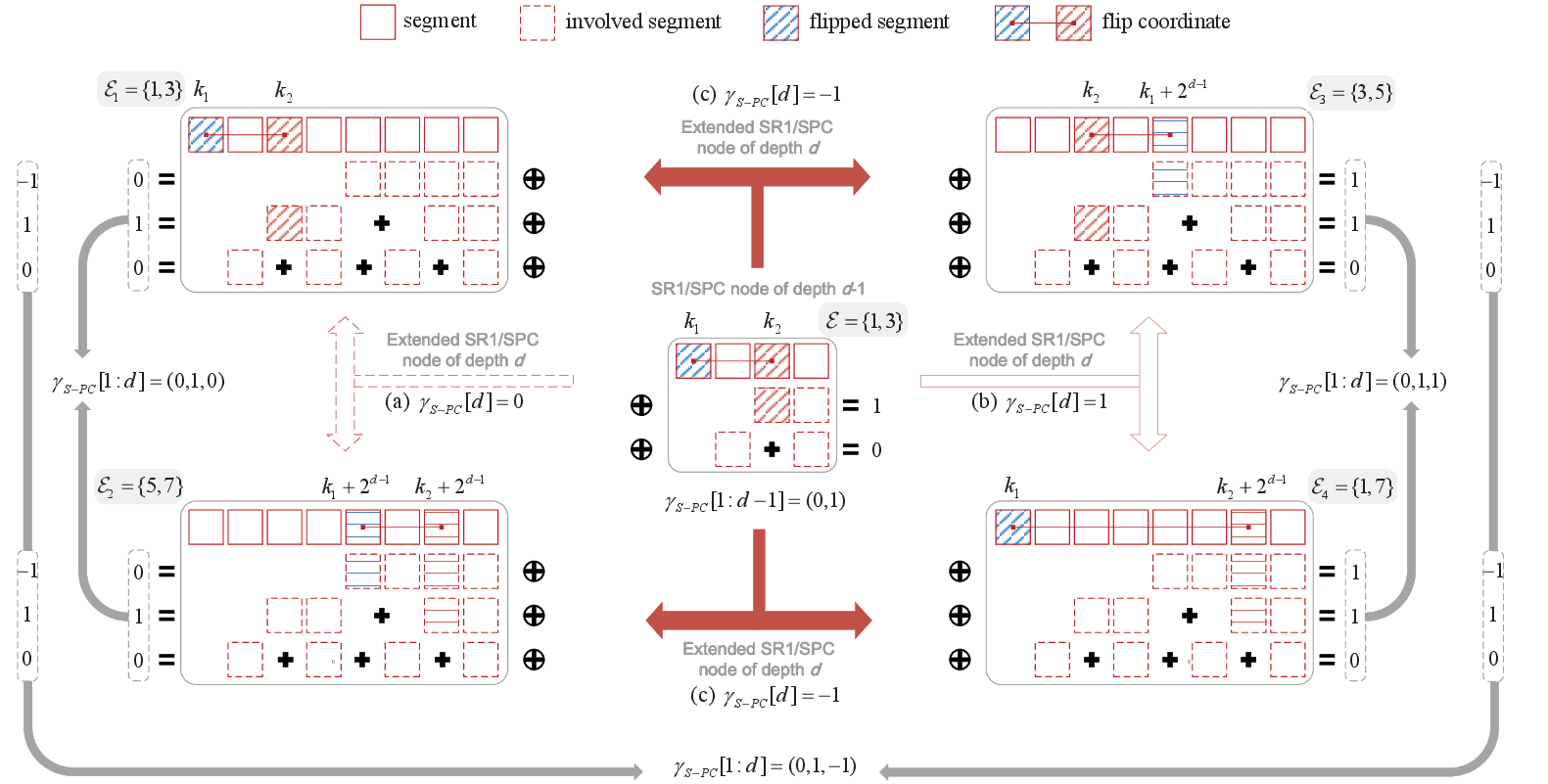}
	\caption{Splitting procedure for constructing the feasible flip coordinates when (a) $\gamma_{\textrm{S-PC}}[d]=0$, (b) $\gamma_{\textrm{S-PC}}[d]=1$ and (c) $\gamma_{\textrm{S-PC}}[d]=-1$.} \label{fig:4}
	\vspace{-0.5cm}
\end{figure*}

The results in Lemma~\ref{lemma:2} indicates a splitting procedure to construct feasible flip coordinates, as shown in Fig.~\ref{fig:4} for an example with $d=3$ and $\gamma_{\textrm{S-PC}}[1:d-1]=(0,1)$. In Fig.~\ref{fig:4}, each box represents a segment, boxes with dashed lines represent the segments involved in a particular S-PC. Besides, a box marked with colour represents a flipped segment, and each flip coordinate is represented by two connected boxes marked with different colours. Note that one can determine whether the parity check of a certain S-PC is correct through the number of involved flipped segments, i.e., flipping an even numbers of segments involved by an S-PC can maintain the corresponding parity check, otherwise the parity check is changed. Given $d=3$ and $\gamma_{\textrm{S-PC}}[1:d-1]=(0,1)$, a feasible flip coordinate can be determined as $\mathcal{E}=\{1,3\}$ according to Lemma~\ref{lemma:1}. It is observed that with $\mathcal{E}=\{1,3\}$, the parity checks of the S-PCs in the depth-$(d-1)$ SR1/SPC node are corrected. Then, $\mathcal{E}$ is split into new feasible flip coordinates, which are determined based on the value of $\gamma_{\textrm{S-PC}}[d]$, as shown in Fig.~\ref{fig:4}~(a), (b) and (c), respectively. If $\gamma_{\textrm{S-PC}}[d]=0$ (Fig.~\ref{fig:4}~(a)), the splitting procedure follows \eqref{eqn:24} in Lemma~\ref{lemma:2}, which leads to $\mathcal{E}_1=\{1,3\}$ and $\mathcal{E}_2=\{5,7\}$. If $\gamma_{\textrm{S-PC}}[d]=1$ (Fig.~\ref{fig:4}~(b)), $\mathcal{E}=\{1,3\}$ is split into $\mathcal{E}_3=\{3,5\}$ and $\mathcal{E}_4=\{1,7\}$ according to \eqref{eqn:25}. Otherwise, for the case of $\gamma_{\textrm{S-PC}}[d]=-1$ (Fig.~\ref{fig:4}~(c)), four flip coordinates $\mathcal{E}_1=\{1,3\}$, $\mathcal{E}_2=\{5,7\}$ $\mathcal{E}_3=\{3,5\}$ and $\mathcal{E}_4=\{1,7\}$ are generated according to \eqref{eqn:26}. For the extended SR1/SPC node of depth $d$, it can also be observed that all the S-PCs are not violated, which validates the correctness of flip coordinates derived from the splitting procedure.

\begin{algorithm}[!t]
	\label{alg2}
	\caption{Generate Flipping Set $\mathcal{F}_{\gamma_{\textrm{S-PC}}}$ for SR1/SPC Nodes}
	\LinesNumbered
	\KwIn{Parities for checking S-PCs $\gamma_{\textrm{S-PC}}[1:p-q]$}
	\KwOut{Flipping set $\mathcal{F}_{\gamma_{\textrm{S-PC}}}$}
	
	$\mathcal{F}_{\gamma_{\textrm{S-PC}}} = \emptyset$\;
	
	\If{$\gamma_{\textrm{S-PC}}[1:d] = (0,\ldots,0,1)$}{
		%			Execute lines 9-12 in Algorithm~\ref{alg2}\;
		\For{$k = \{1,2,\ldots,2^{d-1}\}$}{
			$\mathcal{E} = \{k,k+2^{d-1}\}$\;
			$\mathcal{F}_{\gamma_{\textrm{S-PC}}} = \{\mathcal{F}_{\gamma_{\textrm{S-PC}}},\mathcal{E}\}$\;
		}
		\Return $\mathcal{F}_{\gamma_{\textrm{S-PC}}}$
	}
	\Else{
		Obtain $\mathcal{F}$ associated with $\gamma_{\textrm{S-PC}}[1:d-1]$ using Algorithm~\ref{alg2}\;
		
		\If{$\gamma_{\textrm{S-PC}}[d] = 0$}{
			%				Execute lines 16-20 in Algorithm~\ref{alg2}\;
			\ForEach{$\mathcal{E}$ in $\mathcal{F}$}{
				$\mathcal{E}_1 = \{k_1,k_2\}$, \quad
				$\mathcal{E}_2 = \{k_1+2^{d-1},k_2+2^{d-1}\}$\;
				$\mathcal{F}_{\gamma_{\textrm{S-PC}}} = \{\mathcal{F}_{\gamma_{\textrm{S-PC}}},\mathcal{E}_1, \mathcal{E}_2\}$\;
			}
			\Return $\mathcal{F}_{\gamma_{\textrm{S-PC}}}$
		}
		\ElseIf{$\gamma_{\textrm{S-PC}}[d] = 1$}{
			%				Execute lines 22-26 in Algorithm~\ref{alg2}\;
			\ForEach{$\mathcal{E}$ in $\mathcal{F}$}{
				$\mathcal{E}_3 = \{k_1+2^{d-1},k_2\}$, \quad
				$\mathcal{E}_4 = \{k_1,k_2+2^{d-1}\}$\;
				$\mathcal{F}_{\gamma_{\textrm{S-PC}}} = \{\mathcal{F}_{\gamma_{\textrm{S-PC}}},\mathcal{E}_3, \mathcal{E}_4\}$\;
			}
			\Return $\mathcal{F}_{\gamma_{\textrm{S-PC}}}$
		}
		\Else{
			\ForEach{$\mathcal{E}$ in $\mathcal{F}$}{
				$\mathcal{E}_1=\{k_1,k_2\}$, \qquad $\mathcal{E}_2=\{k_1+2^{d-1},k_2+2^{d-1}\}$\;
				$\mathcal{E}_3=\{k_1+2^{d-1},k_2\}$, \qquad $\mathcal{E}_4=\{k_1,k_2+2^{d-1}\}$\;
				$\mathcal{F}_{\gamma_{\textrm{S-PC}}} = \{\mathcal{F}_{\gamma_{\textrm{S-PC}}},\mathcal{E}_1, \mathcal{E}_2, \mathcal{E}_3, \mathcal{E}_4\}$\;
			}
			\Return $\mathcal{F}_{\gamma_{\textrm{S-PC}}}$
		}
		
	}
\end{algorithm}

Based on Lemmas~\ref{lemma:1} and \ref{lemma:2}, we present in Algorithm~\ref{alg2} the procedure to construct the flipping set $\mathcal{F}_{\gamma_{\textrm{S-PC}}}$ that contains all feasible flip coordinates for SR1/SPC nodes. From Algorithm~\ref{alg2}, it can be observed that there is a boundary condition under which  $\mathcal{F}_{\gamma_{\textrm{S-PC}}}$ can be obtained directly, which is shown in line 2. Lines 3-6 determine $\mathcal{F}_{\gamma_{\textrm{S-PC}}}$ according to Lemma~\ref{lemma:1}. In lines 8-24, a splitting procedure is performed based on Lemma~\ref{lemma:2}, where $\mathcal{F}_{\gamma_{\textrm{S-PC}}}$ is constructed based on $\mathcal{F}$ associated with $\gamma_{\textrm{S-PC}}[1:d-1]$.

Algorithm~\ref{alg2} provides a number of possible flip coordinates for correcting S-PCs, and performing bit-flipping on these bit positions will not violate the P-PCs. Considering the ML rule in \eqref{eqn:6}, we introduce a penalty metric $\lambda_{\mathcal{E}}$ associated with the flip coordinate $\mathcal{E}=\{k_1,k_2,m\}$ ($m \in \{1,2,\ldots,2^q\}$) to derive the optimal codeword, which is given by
\begin{equation}
	\begin{aligned}
		\lambda_{\mathcal{E}} = (1-2 \beta^{\textrm{S-PC}}_{k_1}[m]) \alpha^{\textrm{S-PC}}_{k_1}[m] + (1-2 \beta^{\textrm{S-PC}}_{k_2}[m]) \alpha^{\textrm{S-PC}}_{k_2}[m].
	\end{aligned}
	\label{eqn:27}
\end{equation}
Then, we select the optimal flip coordinate that has the least penalty metric for bit-flipping, i.e.,
\begin{equation}
	\begin{aligned}
		\mathcal{E}_{\textrm{opt}} = \mathop{\arg\min}\limits_{\mathcal{E} \in \mathcal{F}_{\gamma_{\textrm{S-PC}}}} \lambda_{\mathcal{E}}.
	\end{aligned}
	\label{eqn:28}
\end{equation}
Finally, for the codeword obtained from Algorithm~\ref{alg1}, bit-flipping operations are performed on the pair of bit positions determined by $\mathcal{E}_{\textrm{opt}}$, which are shown as follows:
\begin{equation}
	\begin{aligned}
		\beta^{\textrm{S-PC}}_{k_1}[m] = \beta^{\textrm{S-PC}}_{k_1}[m] \oplus 1, \qquad \beta^{\textrm{S-PC}}_{k_2}[m] = \beta^{\textrm{S-PC}}_{k_2}[m] \oplus 1.
	\end{aligned}
	\label{eqn:29}
\end{equation}

\vspace{-0.5em}
\subsection{Decoding of SR1/SPC nodes}
\begin{algorithm}[!t]
	\label{alg3}
	\caption{Decoding of SR1/SPC Nodes}
	\LinesNumbered
	\KwIn{LLR vector $\alpha^{i}_{p}[1:2^p]$ of the SR1/SPC node $\mathcal{N}^{i}_{p}$, node parameters $\operatorname{NS}(p,q,\mathcal{L})$}
	\KwOut{Codeword vector $\beta^{i}_{p}[1:2^p]$ of $\mathcal{N}^{i}_{p}$}
	
	Obtain $\beta^{i}_{p}[1:2^p]$ using Algorithm~\ref{alg1}\;
	
	Calculate $\gamma_{\textrm{S-PC}}[1:d]$ using \eqref{eqn:20}\;
	\If{$\exists t \in \{1,2,\ldots,d\}, \gamma_{\textrm{S-PC}}[t] = 1$}{
		Generate $\mathcal{F}_{\gamma_{\textrm{S-PC}}}$ using Algorithm~\ref{alg2}\;
		\ForEach{$\mathcal{E}$ in $\mathcal{F}_{\gamma_{\textrm{S-PC}}}$}{
			Calculate $\lambda_{\mathcal{E}}$ using \eqref{eqn:27}\;
		}
		Choose $\mathcal{E}_{\textrm{opt}}$ according to \eqref{eqn:28}\;
		Perform bit-flipping operations on $\beta^{i}_{p}[1:2^p]$ using \eqref{eqn:29}\;
	}
	\Return $\beta^{i}_{p}[1:2^p]$
\end{algorithm}

Based on Algorithms~\ref{alg1} and \ref{alg2}, the overall decoding procedure of SR1/SPC nodes is summarized in Algorithm~\ref{alg3}. As can be seen, an SR1/SPC node is first decoded using Algorithm~\ref{alg1} (line 1) to obtain the codeword $\beta^{i}_{p}[1:2^p]$, and at this point the P-PCs of the SR1/SPC node are satisfied. Then, the parity checks of the S-PCs $\gamma_{\textrm{S-PC}}[1:d]$ are calculated (line 2). In case that all the S-PCs are satisfied (depending on the channel conditions), which indicates that the current codeword $\beta^{i}_{p}[1:2^p]$ is optimal, then the whole algorithm is terminated. Otherwise, bit-flipping operations on $\beta^{i}_{p}[1:2^p]$ are required to correct the odd S-PC checks (lines 3-8), i.e., flipping a pair of bits on $\beta^{i}_{p}[1:2^p]$ associated with the flip coordinate. All the feasible flip coordinates are contained in the flipping set which can be determined using Algorithm~\ref{alg2} (line 4). Note that by taking all possible values of $\gamma_{\textrm{S-PC}}[1:d]$ into consideration, all the flipping sets associated with $\gamma_{\textrm{S-PC}}[1:d]$ can be generated offline before decoding. Thus, for a specific vector consists of the S-PC checks, the corresponding flipping set can be determined instantly during decoding. For each feasible flip coordinate in the flipping set, we evaluate its penalty metric according to \eqref{eqn:27} (lines 5-6). Then, the optimal flip coordinate is selected by comparing these penalty metrics (line 7). Finally, the pair of bits indicated by the optimal flip coordinate are flipped to obtain the final codeword vector (line 8). By identifying the proposed SR1/SPC nodes, Algorithm~\ref{alg3} can be incorporated into the existing fast decoding algorithms, which further reduces the decoding latency.

{\color{black}
To select the optimal codeword, the penalty metrics of each flip coordinate from the flipping set should be calculated and compared, resulting in high computational complexity. For a flip coordinate $\mathcal{E}=\{k_1,k_2,m\}$, we say $\mathcal{E}$ is corrupted if $\beta^{\textrm{S-PC}}_{k_1}[m] \not= \operatorname{HD}(\alpha^{\textrm{S-PC}}_{k_1}[m])$ or $\beta^{\textrm{S-PC}}_{k_2}[m] \not= \operatorname{HD}(\alpha^{\textrm{S-PC}}_{k_2}[m])$, which means that only one bit is flipped in the first stage. Otherwise, if $\beta^{\textrm{S-PC}}_{k_1}[m] = \operatorname{HD}(\alpha^{\textrm{S-PC}}_{k_1}[m])$ and $\beta^{\textrm{S-PC}}_{k_2}[m] = \operatorname{HD}(\alpha^{\textrm{S-PC}}_{k_2}[m])$, we say $\mathcal{E}$ is innocent. Given a corrupted flip coordinate $\mathcal{E}_1$ and an innocent one $\mathcal{E}_2$, their penalty metrics can be respectively expressed as $\lambda_{\mathcal{E}_1}=|\alpha_1|-|\alpha_2|$ and $\lambda_{\mathcal{E}_2}=|\alpha_3|+|\alpha_4|$, where $\alpha_i$, $i\in\{1,2,3,4\}$ denotes the LLR value of root node which follows Gaussian distribution. Intuitively, we can see that there is a great probability that $\lambda_{\mathcal{E}_1}<\lambda_{\mathcal{E}_2}$ is satisfied, which indicates that the innocent flip coordinates are less likely to be selected as compared to the corrupted ones. Based on this observation, we present a simplified decoding algorithm to avoid the exhaustive searching of the whole flipping set in line 5 of Algorithm~\ref{alg3}. In particular, at most $2^q$ bits are flipped in the first stage, and the positions of these bits are determined by \eqref{eqn:17}. Then, given a specific $m$, the flip coordinates which have the same segment index as $\rho[m]$ can be determined as potential corrupted ones, where $m=\{1,2,\ldots,2^q\}$. As such, we only need to calculate and compare the penalty metrics of all the potential corrupted flip coordinates. Compared with the plain decoding in Algorithm~\ref{alg3}, this simplified decoding algorithm is more friendly for hardware implementation since some unnecessary computations are released (see Table~\ref{tab4}), and it is important to mention that this simplification technique will not degrade the decoding performance (see Fig.~\ref{fig:6}).
}

\vspace{-0.5em}
\section{Simulation Results}
\label{section:5}
In this section, the node distribution, average decoding latency, computational complexity, and error-correction performance of the proposed decoding algorithms are investigated, and comparison with the state-of-the-art fast SC decoding algorithms is provided. We consider polar codes with length $N \in \{512,1024,4096\}$. When $N \leq 1024$, polar codes are constructed according to the 5G standard \cite{5Gstandard}, while for $N = 4096$, the Gaussian Approximation (GA) method \cite{Trifonov2012Efficient} is utilized for code construction. 

\vspace{-0.5em}
{\color{black}
\subsection{Node Distribution}

\begin{table*}[tb] \scriptsize\color{black}
	\caption{Distributions of EG-PC and SR1/SPC Nodes}
	\centering
	\begin{threeparttable}
		\setlength{\tabcolsep}{1.5mm}{
		\begin{tabular}{cc|cccccc|cccccc|cccccc|c|cccccc|c}\hline
			\multirow{3}{*}{$N$} & \multirow{3}{*}{$R$} &
			\multicolumn{6}{c|}{{SR1}} & \multicolumn{6}{c|}{SSPC} & \multicolumn{6}{c|}{SR1/SPC (other cases)} & \multirow{3}{*}{Total} & \multicolumn{6}{c|}{EG-PC} & \multirow{3}{*}{Total} \\ \cline{3-20} \cline{22-27} & & \multicolumn{6}{c|}{Node Depth $d$} & \multicolumn{6}{c|}{Node Depth $d$} & \multicolumn{6}{c|}{Node Depth $d$} & & \multicolumn{6}{c|}{Node Depth $d$} & \\ \cline{3-20} \cline{22-27} 
			& & 1 & 2 & 3 & 4 & 5 & $\geq$ 6 & 1 & 2 & 3 & 4 & 5 & $\geq$ 6 & 1 & 2 & 3 & 4 & 5 & $\geq$ 6 & & 1 & 2 & 3 & 4 & 5 & $\geq$ 6 & \\ \hline			
			& 1/6 & 0 & 1 & 0 & 1 & 0 & 0  & 0 & 1 & 0 & 0 & 0 & 0  & 0 & 0 & 0 & 0 & 0 & 0  & 3  & 0 & 1 & 0 & 1 & 0 & 0 & 2 \\
			& 1/3 & 0 & 1 & 1 & 0 & 0 & 0  & 0 & 1 & 1 & 0 & 0 & 0  & 0 & 0 & 1 & 1 & 0 & 0  & 6  & 0 & 1 & 1 & 0 & 0 & 0 & 2 \\
			512\tnote{1} & 1/2 & 0 & 1 & 1 & 0 & 0 & 0  & 1 & 1 & 0 & 0 & 0 & 0  & 0 & 1 & 1 & 1 & 0 & 0  & 7  & 0 & 1 & 1 & 0 & 0 & 0 & 2 \\
			& 2/3 & 0 & 0 & 1 & 0 & 1 & 0  & 0 & 2 & 1 & 0 & 0 & 0  & 0 & 0 & 0 & 2 & 0 & 0  & 7  & 0 & 0 & 1 & 0 & 1 & 0 & 2 \\
			& 5/6 & 0 & 0 & 0 & 1 & 0 & 0  & 0 & 1 & 1 & 0 & 0 & 1  & 0 & 1 & 0 & 0 & 0 & 0  & 5  & 0 & 0 & 0 & 1 & 0 & 0 & 1 \\ \hline			
			& 1/6 & 0 & 2 & 1 & 0 & 0 & 1  & 1 & 3 & 4 & 2 & 0 & 0  & 0 & 0 & 0 & 1 & 0 & 0  & 15  & 0 & 2 & 1 & 0 & 0 & 1 & 4 \\
			& 1/3 & 0 & 2 & 2 & 1 & 0 & 0  & 1 & 8 & 6 & 3 & 1 & 0  & 0 & 0 & 2 & 0 & 1 & 1  & 28  & 0 & 2 & 2 & 1 & 0 & 0 & 5 \\
			4096\tnote{2} & 1/2 & 0 & 1 & 3 & 1 & 1 & 0  & 0 & 7 & 8 & 2 & 0 & 2  & 0 & 1 & 2 & 1 & 1 & 0  & 30  & 0 & 1 & 3 & 1 & 1 & 0 & 5 \\
			& 2/3 & 0 & 5 & 4 & 0 & 0 & 1  & 4 & 7 & 4 & 4 & 0 & 0  & 0 & 0 & 1 & 1 & 2 & 1  & 34  & 0 & 5 & 4 & 0 & 0 & 1 & 9 \\
			& 5/6 & 0 & 4 & 2 & 2 & 0 & 1  & 1 & 6 & 1 & 0 & 0 & 0  & 0 & 1 & 1 & 1 & 2 & 2  & 24  & 0 & 4 & 2 & 2 & 0 & 1 & 8 \\ \hline
		\end{tabular}}
		\begin{tablenotes}  
			\footnotesize
			\item[1] Constructed according to the 5G standard \cite{5Gstandard}.    
			\item[2] Constructed by the GA method \cite{Trifonov2012Efficient}.    
		\end{tablenotes} 
	\end{threeparttable}
	\label{tab2}
	\vspace{-0.5cm}
\end{table*}

Table~\ref{tab2} shows the numbers of EG-PC and SR1/SPC nodes contained in a particular polar code with different code lengths $N$ and code rates $R$. It can be seen that generally, the number of SR1/SPC nodes increases with the increasing of the code length $N$, i.e., more latency savings can be achieved for longer polar codes. When $R>1/2$, it is observed that the total number of SR1/SPC nodes is smaller, but the node depth is larger as the node length $N$ increases. This is because as the code rate increases, the proposed SR1/SPC nodes tend to be located at higher levels of the decoding tree, and thus the number of such nodes is smaller. Besides, although the proposed SR1 node a generalized version of the EG-PC node, they have almost the same distribution in a practical polar code, which means that the source node in the SR1 node is usually Rate-0 or REP. However, other types of SR1/SPC nodes are shown to be more frequently distributed as compared to the EG-PC nodes, leading to significant overall latency reduction. We also note that the implementation complexity can be significantly reduced by using the proposed decoder, since we do not need to develop a dedicated decoder for each special node type shown in Table~\ref{tab1}. In particular, 5G polar codes with $N = 128$ can be represented by only SR0/REP and SR1/SPC nodes.
}
\vspace{-0.5em}
\subsection{Decoding Latency and Computational Complexity Analysis}
In this subsection, we measure the decoding latency of Algorithm~\ref{alg3} by the number of required time steps, where the same unlimited-resource assumptions as in \cite{Alamdar2011Simplified,Hanif2017Fast,Zheng2021Threshold} are utilized, which are shown as follows:
\begin{itemize}
	\item Addition/subtraction of real numbers can be performed in one time step.
	\item Bit operations can be carried out instantly.
	\item Wagner decoding consumes one time step.
\end{itemize}

Let $T_1$ and $T_2$ denote respectively the corresponding required time steps of the first stage and second stage in Algorithm~\ref{alg3}. First, we analyse $T_1$, i.e., the number of required time steps to perform Algorithm~\ref{alg1}. As can be seen, the decoding in the first stage depends on the source node structure. In case that the source node is Rate-0, its codeword can be obtained immediately. Otherwise, calculating the LLRs and decoding the source node require one and $T_{\textrm{SN}}$ time steps, respectively, where the value of $T_{\textrm{SN}}$ is determined by the decoding algorithm for the source node. Moreover, applying Wagner decoding to the root node requires one time step. Therefore, $T_1$ can be expressed as
\begin{equation}
	\begin{aligned}
		T_1 = \begin{cases}
			1, & \textrm{if the source node is Rate-0} \\
			T_{\textrm{SN}} + 2, & \textrm{otherwise}
		\end{cases}.
	\end{aligned}
	\label{eqn:30}
\end{equation}
Subsequently, we can see that calculating the parity checks of the S-PCs in line 2 of Algorithm~\ref{alg3} can be performed instantaneously since it only involves XOR bit operations. Note that the flipping set can be generated offline, thus conducting line 4 incurs no latency. The calculation of the penalty metrics in lines 5-6 only involves addition operations and can be executed in parallel, which requires one time step. The following selection of the optimal flip coordinate in line 7 is based on a comparison tree, and thus consuming another one time step. Finally, performing bit-flipping operations (line 8) requires no time step, since only XOR bit operations are involved as shown in \eqref{eqn:29}. As a result, $T_2$ is given by 
\begin{equation}
	\begin{aligned}
		&T_2
		=\begin{cases}
			2, & \textrm{if}\ \exists t \in \{1,2,\ldots,d\}, \gamma_{\textrm{S-PC}}[t] = 1 \\
			0, & \textrm{otherwise}
		\end{cases}.
	\end{aligned}
	\label{eqn:31}
\end{equation}
As can be seen, $T_2$ is a variable depending on $\gamma_{\textrm{S-PC}}[1:d]$, which is further dependent on the channel conditions. To summarize, the total numbers of time steps required to decode SR1, SSPC and SR1/SPC node are respectively given by
\begin{equation}
	\begin{aligned}
		T_{\textrm{SR1}} = T_1, \qquad T_{\textrm{SSPC}} = T_{\textrm{SR1/SPC}} = T_1 + T_2.
	\end{aligned}
	\label{eqn:32}
\end{equation}
Furthermore, the lower and upper bounds of $T_{\textrm{SR1/SPC}}$, denoted by $T_{\textrm{LB}}$ and $T_{\textrm{UB}}$, respectively, are given by
\begin{equation}
	\begin{gathered}
		T_{\textrm{LB}} = T_1 + \operatorname{min}(T_2) = T_1, \qquad
		T_{\textrm{UB}} = T_1 + \operatorname{max}(T_2) = T_1 + 2.
	\end{gathered}
	\label{eqn:33}
\end{equation}
Consequently, for a given polar code, the overall, minimum and maximum numbers of decoding time steps can be calculated based on \eqref{eqn:32} and \eqref{eqn:33}.

\begin{table*}[tb] \scriptsize\color{black}
	\caption{Number of Required Time steps of different decoding algorithms}
	\centering
	\begin{threeparttable}
		\begin{tabular}{cccccccccccccc}\hline
			\multirow{3}{*}{$N$} & \multirow{3}{*}{$R$} & \multirow{3}{*}{FSSC \cite{Sarkis2014Fast}} & \multirow{3}{*}{HFSC1\tnote{1}} & \multirow{3}{*}{HFSC2\tnote{2}} &
			\multirow{3}{*}{HFSC3\tnote{3}} & \multicolumn{7}{c}{SN-FSC\tnote{4}} & \multirow{3}{*}{SN-RFSC\tnote{5}}\\ \cline{7-13}
			& & & & & & \multicolumn{5}{c}{$E_b/N_0$ (dB)} & \multirow{2}{*}{LB} & \multirow{2}{*}{UB} & \\
			& & & & & & 0.0 & 1.0 & 2.0 & 3.0 & 4.0 & & & \\ \hline
			& 1/6 & 89 & 63 & 55 & 36 & 34.86 & 34.29 & 34.06 & 34.01 & 34.00 & 34 & 36 & 34 \\
			& 1/3 & 128 & 85 & 86 & 63 & 47.94 & 44.87 & 43.39 & 43.05 & 43.00 & 43 & 51 & 43 \\
			512 & 1/2 & 126 & 89 & 101 & 70 & 59.27 & 57.32 & 54.63 & 54.09 & 54.01 & 54 & 64 & 54 \\
			& 2/3 & 129 & 87 & 103 & 74 & 58.12 & 56.87 & 53.54 & 50.50 & 50.03 & 50 & 60 & 50 \\
			& 5/6 & 88 & 64 & 72 & 58 & 38.40 & 37.88 & 37.84 & 35.14 & 32.61 & 32 & 40 & 32 \\ \hline
			
			& 1/6 & 446 & 347 & 339 & 234 & 197.75 & 185.37 & 184.03 & 184.00 & 184.00 & 184 & 206 & 184 \\
			& 1/3 & 639 & 485 & 525 & 368 & 289.28 & 262.85 & 255.07 & 255.00 & 255.00 & 255 & 301 & 255 \\
			4096 & 1/2 & 674 & 508 & 547 & 402 & 309.26 & 301.29 & 271.53 & 271.01 & 308.00 & 271 & 319 & 271 \\
			& 2/3 & 590 & 434 & 474 & 354 & 273.04 & 273.24 & 248.48 & 237.15 & 237.01 & 237 & 285 & 237 \\
			& 5/6 & 409 & 282 & 321 & 244 & 185.28 & 184.96 & 184.68 & 173.33 & 163.23 & 163 & 193 & 163 \\ \hline
		\end{tabular}
		\begin{tablenotes}  
			\footnotesize
			\item[1] FSSC + Type I-V \cite{Hanif2017Fast}.
			\item[2] FSSC + G-REP \cite{Condo2018Generalized} + G-PC \cite{Condo2018Generalized}.  
			\item[3] FSSC + EG-PC \cite{Zheng2021Threshold} + SR0/REP \cite{Zheng2021Threshold}.
			\item[4] FSSC + SR0/REP \cite{Zheng2021Threshold} + SR1/SPC.
			\item[5] FSSC + SR0/REP \cite{Zheng2021Threshold} + RG-PC \cite{Condo2018Generalized}.
		\end{tablenotes} 
	\end{threeparttable}
	\label{tab3}
	\vspace{-0.5cm}
\end{table*}

In Table~\ref{tab3}, we compare the time steps required by the state-of-the-art fast SC decoding algorithms in \cite{Sarkis2014Fast,Hanif2017Fast,Condo2018Generalized,Zheng2021Threshold} and the proposed fast SC decoding algorithm, where two code lengths $N \in \{512, 4096\}$ and five code rates $R \in \{1/6, 1/3, 1/2, 2/3, 5/6\}$ are considered. The FSSC decoding algorithm proposed in \cite{Sarkis2014Fast}, which takes Rate-0, Rate-1, REP and SPC nodes into consideration, is selected as the baseline. Furthermore, the decoding of Type I-V nodes \cite{Hanif2017Fast}, G-REP and G-PC nodes \cite{Condo2018Generalized}, and SR0/REP nodes \cite{Zheng2021Threshold} are progressively applied to the baseline algorithm, leading to three hybrid fast SC decoders, namely HFSC1, HFSC2 and HFSC3, respectively. To evaluate the effects of the proposed SR1/SPC nodes on the decoding speed, the decoding of Rate-0, Rate-1, REP, SPC and SR0/REP nodes is combined with the decoding of the proposed SR1/SPC node in Algorithm~\ref{alg3}, and the resulting fast SC decoder with sequence nodes, abbreviated as SN-FSC, is used for comparison. The decoding latency of the proposed SN-FSC decoder is measured by the average number of time steps at $E_b/N_0 = \{0.0,1.0,2.0,3.0,4.0\}$ dB, since the required time steps of decoding SR1/SPC nodes $T_{\textrm{SR1/SPC}}$ is an variable affected by the channel condition, while those of the baseline algorithms are fixed. Besides, the minimum and maximum numbers of required time steps are also presented, which are denoted by `LB' and `UB', respectively. {\color{black}In order to gain further speedup, we also present a relaxed fast SC decoder with sequence nodes, namely the SN-RFSC decoder, where the SR1/SPC nodes are decoded as RG-PC nodes by ignoring the frozen bits in the SR1/SPC descendant nodes \cite{Condo2018Generalized}.} As can be seen, the proposed decoder requires fewer time steps with respect to the other decoders for all considered $N$ and $R$. For different values of $N$, it can be observed that the gain brought by decoding the proposed SR1/SPC nodes increases with $R$. In particular, when $N=512,R=5/6$ and $E_b/N_0=4.0$ dB, the proposed decoder can save up to 62.9\% and 43.8\% time steps as compared to the FSSC decoder and the HFSC2 decoder. However, the speedup gain is limited when the code length is short and the code rate is low, which is mainly because the numbers of SR1/SPC nodes in these cases are small. Besides, the gain brought by SR1/SPC nodes is larger when $E_b/N_0$ is higher, and when $E_b/N_0=4.0$ dB the decoding latency $T_{\textrm{SR1/SPC}}$ tends to approach the lower bound $T_{\textrm{SR1}}$. This implies that the S-PCs of the SR1/SPC nodes are automatically satisfied under good channel conditions. {\color{black}It can also be observed that the SN-RFSC decoder has the minimum number of required time steps, i.e., it can achieve the decoding latency lower bound of the SN-FSC decoder. However, as can be seen from Fig.~\ref{fig:5} and Fig.~\ref{fig:6}, employing this relaxed decoder also incurs severe performance loss as compared to the conventional SC decoder, thus it may not be suitable for practical implementation. }

\begin{table*}[tb] \scriptsize\color{black}
	\caption{Computational Complexity Comparison of Different Decoders for $\mathcal{P}(1024,512)$}
	\centering
	\begin{threeparttable}
		\begin{tabular}{ccccccccccccc}\hline
			\multirow{3}{*}{ACS Operations} & \multirow{3}{*}{SC} & \multirow{3}{*}{FSSC} & \multirow{3}{*}{HFSC1} & \multirow{3}{*}{HFSC2} & \multirow{3}{*}{HFSC3} & \multicolumn{7}{c}{SN-FSC} \\ \cline{7-13}
			& & & & & & \multicolumn{5}{c}{$E_b/N_0$ (dB)} & \multirow{2}{*}{LB} & \multirow{2}{*}{UB} \\
			& & & & & & 0.0 & 1.0 & 2.0 & 3.0 & 4.0 & & \\ \hline			
			\multirow{2}{*}{Add} & \multirow{2}{*}{5120} & \multirow{2}{*}{3110} & \multirow{2}{*}{3056} & \multirow{2}{*}{3110} & \multirow{2}{*}{5379} & 5244.36 & 5144.23 & 4958.23 & 4930.72 & 4928.19 & 4928 & 5440 \\
			& & & & & & 4999.12 & 4976.62 & 4936.26 & 4928.95 & 4928.09 & 4928 & 5056 \\ \hline
			\multirow{2}{*}{Compare} & \multirow{2}{*}{5120} & \multirow{2}{*}{2742} & \multirow{2}{*}{2612} & \multirow{2}{*}{2604} & \multirow{2}{*}{2220} & 2216 & 2216 & 2216 & 2216 & 2216 & 2216 & 2216 \\
			& & & & & & 2216 & 2216 & 2216 & 2216 & 2216 & 2216 & 2216 \\ \hline
			\multirow{2}{*}{Sort\tnote{1}} & \multirow{2}{*}{0} & \multirow{2}{*}{364} & \multirow{2}{*}{436} & \multirow{2}{*}{384} & \multirow{2}{*}{480} & 970.37 & 870.23 & 684.23 & 656.72 & 654.19 & 654 & 1166 \\
			& & & & & & 725.12 & 702.62 & 662.26 & 654.95 & 654.09 & 654 & 782 \\ \hline
			\multirow{2}{*}{Total} & \multirow{2}{*}{10240} & \multirow{2}{*}{6216} & \multirow{2}{*}{6104} & \multirow{2}{*}{6098} & \multirow{2}{*}{8079} & 8430.71 & 8230.45 & 7858.46 & 7803.44 & 7798.37 & 7798 & 8822 \\
			& & & & & & 7940.24 & 7895.24 & 7814.52 & 7799.91 & 7798.19 & 7798 & 8054 \\ \hline

		\end{tabular}
		\begin{tablenotes}  
		\footnotesize
		\item[1] This metric is obtained by counting the total number of data that is required to be sorted.
		\end{tablenotes} 
	\end{threeparttable}
	\label{tab4}
	\vspace{-0.5cm}
\end{table*}

{\color{black}The above latency analysis of Algorithm~\ref{alg3} is based on the unlimited-resource assumption where operations can be parallelly executed in a single time step. In the other extreme, i.e., the case of a fully serial implementation, the number of arithmetic (including add, compare and sort (ACS)) operations can also be employed to measure the average computational complexity. Table~\ref{tab4} compares the number of ACS operations required by the considered decoders, for $\mathcal{P}(1024,512)$. For the proposed SN-FSC decoder, the two numbers on each row/column in Table~\ref{tab4} provide the numbers of ACS operations required by the plain decoder and the simplified decoder, respectively. As can be seen, the decoding of SR0/REP node \cite{Zheng2021Threshold} mainly incurs addition operations, which is caused by LLR calculation of a list of repetition sequence patterns, and it was shown in \cite{Zheng2020Implementation} that adopting the SR0/REP node increases the computational resource consumption. Similar to the SR0/REP node, the decoding of the proposed SR1/SPC nodes requires more sorters to correct the P-PCs and S-PCs, which also increases the number of overall ACS operations. However, we can observe from Table~\ref{tab4} that the ACS operations required by the proposed plain SN-FSC decoder is only slightly increased as compared to other fast SC decoders and still significantly lower than the original SC decoder. In addition, with the proposed simplification technique, the SN-FSC decoder can achieve lower computational complexity as compared to the HFSC3 decoder. In general, fast decoding of the proposed sequence nodes achieves significant decoding speedup at the cost of slightly increased hardware resource consumption, which is very promising.}

\vspace{-0.5em}
{\color{black}
\subsection{Error-Correction Performance Comparison}

\begin{figure}[!t]\color{black}
	\centering
	\setlength{\abovecaptionskip}{-0.1cm} 
	\includegraphics[width=0.45\textwidth]{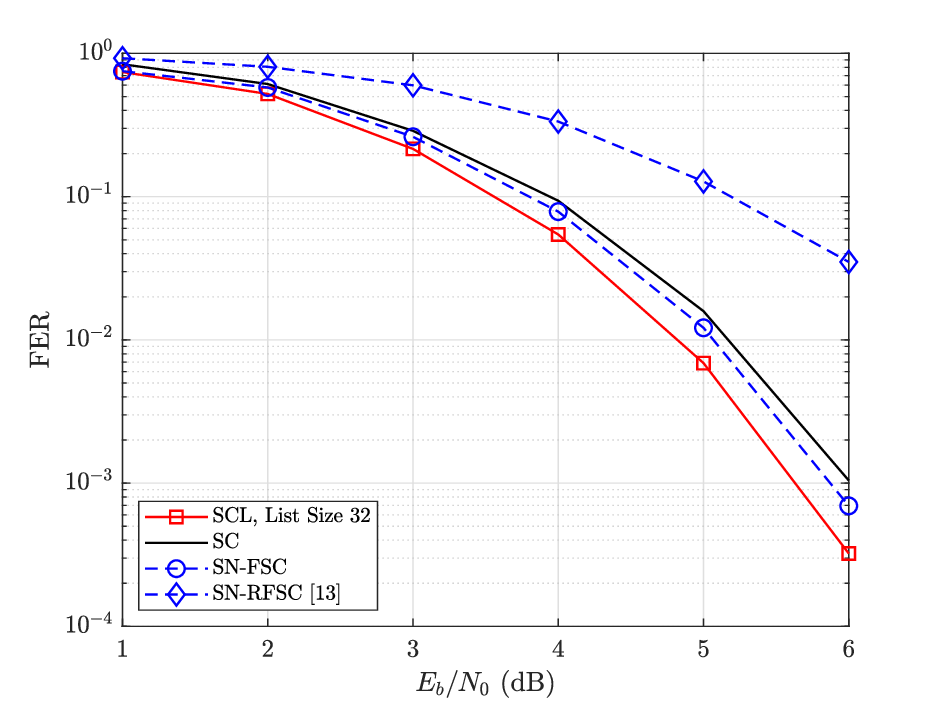}
	\caption{FER performance comparison of the SC, SCL, SN-FSC and SN-RFSC decoders, for an SSPC node $\operatorname{NS}(6,2,\{2,3,4,5\})$ with Rate-0 source node.}
	\label{fig:5}
	\vspace{-0.5cm}
\end{figure}

To avoid costly computations, an SCL decoder with list size 32 is implemented to get close-to-ML decoding performance. In Fig.~\ref{fig:5}, we compare the frame error rate (FER) performance of different decoders, when decoding an SSPC polar subcode $\operatorname{NS}(6,2,\{2,3,4,5\})$ that can be find in a GA polar code with $N=4096$. It can be seen that although the performance of the proposed decoding algorithm is inferior to that achieved by the SCL decoder (with list size 32), it still outperforms the conventional SC decoding. This means the proposed decoding algorithms can achieve near-ML performance for the SR1/SPC nodes. However, ignoring the S-PCs, as done in \cite{Condo2018Generalized}, causes severe performance degradation (larger than 1 dB). As a result, it is vital to take all parity constraints into consideration in order to preserve the decoding performance.

\begin{figure}[!t]\color{black}
	\centering
	\setlength{\abovecaptionskip}{-0.1cm} 
	\includegraphics[width=0.45\textwidth]{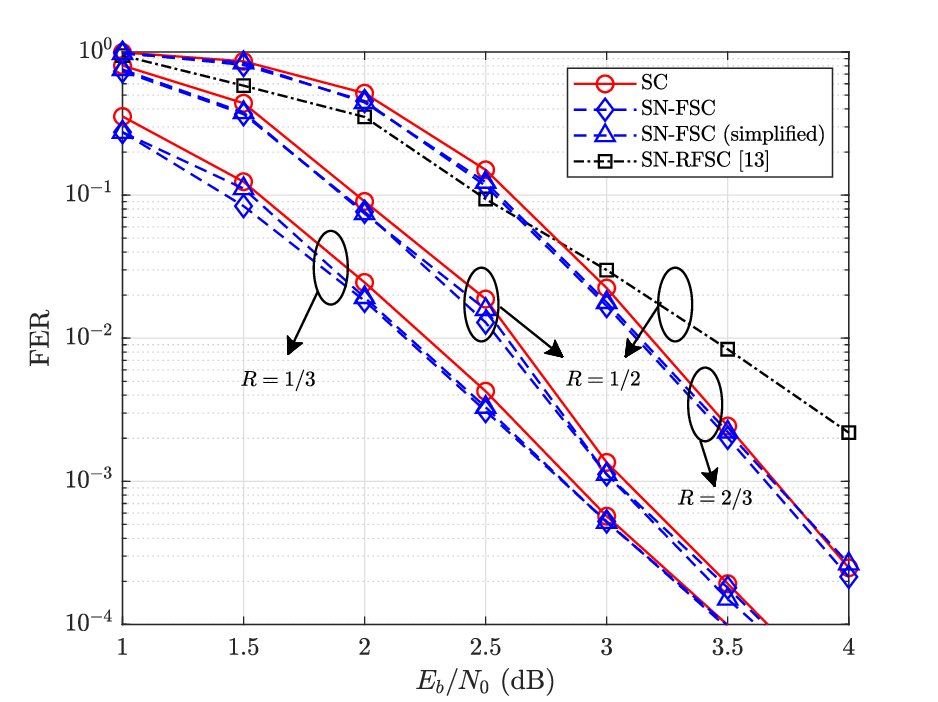}
	\caption{FER performance comparison of the SC, SN-FSC and SN-RFSC decoders, for 5G polar codes with length $N=1024$ and $R=\{1/3,1/2,2/3\}$.}
	\label{fig:6}
	\vspace{-0.5cm}
\end{figure}

A similar phenomenon can also be observed from Fig.~\ref{fig:6}, where the proposed decoder is shown to outperform the conventional SC decoder for a $N=1024$ polar code with $R \in \{1/3,1/2,2/3\}$. Note that this is an important advantage of the proposed decoder since the performance of previous fast SC decoders cannot exceed that of the conventional SC decoder. Besides, Fig.~\ref{fig:6} also shows that the employment of the proposed simplification technique incurs no performance loss, yet with reduced computational complexity.
}

\vspace{-0.5em}
\section{Conclusions}
\label{section:6}
In this work, we introduced a new class of multi-node information and frozen bit patterns, namely the SR1/SPC nodes, and proposed the corresponding fast SC decoding algorithm. The proposed SR1/SPC node can be identified by a sequence of Rate-1 or SPC nodes, and it provides a unified description of a wide variety of special nodes. Moreover, the decoding rules for the proposed SR1/SPC nodes were derived to facilitate fast SC decoding, with which a higher degree of parallelism can be obtained. The proposed decoding algorithm was compared with the state-of-the-art fast SC decoding algorithms, in terms of decoding latency and error-correction performance. Simulation results showed that by combining the existing special nodes with the proposed SR1/SPC nodes, we can achieve much lower decoding latency as well as noticeable performance improvement. Besides, we note that the proposed SR1/SPC nodes are very promising as they can also be adopted in the SCL decoder, which is left for future work.
\vspace{-0.5em}
\appendices
\section{Proof of Theorem~\ref{theorem:1}}
\label{appendex:1}
Given the codeword of each node at level $q$, the codeword of the root node at level $p$ can be derived based on the bit recursion formula in \eqref{eqn:4}, i.e., 
\begin{equation}
	\begin{aligned}
		\beta^{i}_{p}[1:2^p] = \bm{\beta}_{q} (\bm{F}^{\otimes d} \otimes \bm{I}_{2^q}),
	\end{aligned}
	\label{aeqn:1}
\end{equation}
where $\bm{\beta}_{q} = (\beta^{L_q}_{q}[1:2^q], \beta^{L_q+1}_{q}[1:2^q], \ldots, \beta^{R_q}_{q}[1:2^q])$. Then, using the identity $(\bm{A} \otimes \bm{B}) \otimes \bm{C} = \bm{A} \otimes (\bm{B} \otimes \bm{C})$ with $\bm{A} = \bm{B} = \bm{F}$ and $\bm{C} = \bm{I}_{2^q}$, \eqref{aeqn:1} can be rewritten as
\begin{equation}
	\begin{aligned}
		\beta^{i}_{p}[1:2^p] =& \bm{\beta}_{q} (\bm{F}^{\otimes d-1} \otimes (\bm{F} \otimes \bm{I}_{2^q})) 
		= \bm{\beta}_{q} \bigg(\bm{F}^{\otimes d-1} \otimes \begin{pmatrix} \bm{I}_{2^q} & \bm{0}_{2^q} \\ \bm{I}_{2^q} & \bm{I}_{2^q} \end{pmatrix}\bigg).
	\end{aligned}
	\label{aeqn:2}
\end{equation}
Repeating the above procedures $d-1$ times leads to
\begin{equation}
	\begin{aligned}
		\beta^{i}_{p}[1:2^p] = \bm{\beta}_{q} \bm{G}^\prime_{2^d},
	\end{aligned}
	\label{aeqn:3}
\end{equation}
where $\bm{G}^\prime_{2^d}$ is obtained by replacing the 0 and 1 elements in $\bm{G}_{2^d}$ by $\bm{0}_{2^q}$ and $\bm{I}_{2^q}$, respectively. Then, for any $k \in \{1,2,\ldots,2^q\}$, we have
\begin{equation}
	\begin{aligned}
		\bigoplus\limits_{j=1}^{2^d} \beta^{i}_{p}[(j-1)2^q+k] & =  \bigoplus\limits_{j=1}^{2^d}(\bm{\beta}_{q}  \bm{G}^\prime_{2^d})_{(j-1)2^q+k} = \bm{\beta}_{q} \bigoplus\limits_{j=1}^{2^d} (\bm{G}^\prime_{2^d})_{(j-1)2^q+k} \\
		&\overset{(a)}{=} \bm{\beta}_{q} ((\bm{I}_{2^q}, \overbrace{\bm{0}_{2^q}, \ldots, \bm{0}_{2^q}}^{2^d-1})^T)_k = \bm{\beta}_{q} (\overbrace{0,\ldots,0}^{k-1},1,\overbrace{0,\ldots,0}^{2^p-k})^T \\
		& = \beta_{q}[k] \overset{(b)}{=} \beta^{L_q}_{q}[k],
	\end{aligned}
	\label{aeqn:4}
\end{equation}
where $(a)$ is based on the property of the generator matrix, i.e., $\bigoplus_{j=1}^{2^d}(\bm{G}_{2^d})_j = (1,0,\ldots,0)^T$, and $(b)$ is due to the definition of $\bm{\beta}_{q}$. This thus completes the proof.
\vspace{-0.5em}
\section{Proof of Theorem~\ref{theorem:2}}
\label{appendex:2}
Similar to the proof of Theorem~\ref{theorem:1}, the following equality can be obtained:
\begin{equation}
	\begin{aligned}
		\bigoplus\limits_{j=1}^{2^{p-r-1}} \beta^{i}_{p}[(2j-1)2^r+k]  & \overset{(a)}{=} \bm{\beta}_{r} \bigoplus\limits_{j=1}^{2^{p-r-1}} (\bm{G}^\prime)_{(2j-1)2^r+k} \overset{(b)}{=} \bm{\beta}_{r} ((\bm{0}_{2^r}, \bm{I}_{2^r}, \overbrace{\bm{0}_{2^r}, \ldots, \bm{0}_{2^r}}^{2^{p-r}-2})^T)_k \\
		&= \bm{\beta}_{r} (\overbrace{0,\ldots,0}^{2^r+k-1},1,\overbrace{0,\ldots,0}^{2^p-2^r-k})^T = \beta_{q}[2^r+k] = \beta^{L_r+1}_{r}[k],
	\end{aligned}
	\label{aeqn:5}
\end{equation}
where $(a)$ is based on \eqref{aeqn:3} and $(b)$ is because $\bigoplus_{j=1}^{2^{p-r-1}}(\bm{G}_{2^{p-r}})_{2j} = (0,1,\ldots,0)^T$. Since $\mathcal{N}^{L_r+1}_{r}$ is an SPC node, its corresponding codeword should satisfy \cite{Sarkis2014Fast}
\begin{equation}
	\begin{aligned}
		\bigoplus\limits_{k=1}^{2^{r}} \beta^{L_r+1}_{r}[k] = 0. 
	\end{aligned}
	\label{aeqn:6}
\end{equation}		
By replacing $\beta^{L_r+1}_{r}[k]$ on the left-hand-side of \eqref{aeqn:6} by \eqref{aeqn:5}, \eqref{eqn:13} can be readily proved.
\vspace{-0.5em}
\section{Proof of Theorem~\ref{theorem:3}}
\label{appendex:3}
First, we prove the equivalence of segments $k$ and $k+\mu2^t$, $\mu = \{\lceil -k/2^t \rceil,\ldots,\lfloor (2^d-k)/2^t \rfloor\} \setminus \{0\}$. Starting from $\mu=1$, if segment $k$ is involved in the $t$-th S-PC, i.e., $k \in \mathcal{I}_{t}$, the range of $k+2^t$ can be accordingly determined by
\begin{equation*}
	\begin{aligned}
		k \in \mathcal{I}_{t} & \Rightarrow (2j-1)2^{t-1}+1+2^t \leq k+2^t \leq j2^t+2^t \\
		&\Rightarrow (2(j+1)-1)2^{t-1}+1 \leq k+2^t \leq (j+1)2^t \\
		&\Rightarrow k+2^t \in \mathcal{I}_{t}.
	\end{aligned}
\end{equation*}
By viewing $j+1$ as a new $j$, it can be observed from \eqref{eqn:21} that segment $k+2^t$ is also involved in the $t$-th S-PC. Likewise, if segment $k$ is not involved in the $t$-th S-PC, i.e., $k \in \mathcal{I}_{t}^c$, the range of $k+2^t$ can be determined based on \eqref{eqn:22}, given by
\begin{equation*}
	\begin{aligned}
		k \in \mathcal{I}_{t}^c & \Rightarrow (j-1)2^t+1+2^t \leq k+2^t \leq (2j-1)2^{t-1}+2^t \\
		&\Rightarrow j2^t+1 \leq k+2^t \leq (2(j+1)-1)2^{t-1} \\
		&\Rightarrow k+2^t \in \mathcal{I}_{t}^c,
	\end{aligned}
\end{equation*}
which shows that segment $k+2^t$ is not related to the $t$-th S-PC. Therefore, segments $k$ and $k+2^t$ are equivalent for $j=1$. Then, following a similar procedure, it can be proved that segment $k+2^t$ is equivalent to segment $k+2\times2^t$ and the latter is further equivalent to $k+3\times2^t$, etc. To sum up, we can prove that segment $k$ is equivalent to segments $k+\lceil -k/2^t \rceil2^t,\ldots,k-2^t,k+2^t,\ldots,k+\lfloor (2^d-k)/2^t \rfloor2^t$, which thus proves that segments $k+\mu2^t$, $\mu = \{\lceil -k/2^t \rceil,\ldots,\lfloor (2^d-k)/2^t \rfloor\}$ are equivalent with each other.

Then, we prove that segments $k$ and $k+(2\nu-1)2^{t-1}$ are inequivalent ($\nu \in \{\lceil -k/2^t+1/2 \rceil,\ldots,\lfloor (2^d-k)/2^t+1/2 \rfloor\}$). Let $\nu=1$, depending on whether segment $k$ is involved in the $t$-th S-PC or not, the ranges of $k+2^{t-1}$ can be respectively obtained as follows based on \eqref{eqn:21} and \eqref{eqn:22}:
\begin{equation*}
	\begin{aligned}
		k \in \mathcal{I}_{t} & \Rightarrow (2j-1)2^{t-1}+1+2^{t-1} \leq k+2^{t-1} \leq j2^t+2^{t-1} \\
		&\Rightarrow j2^t+1 \leq k+2^{t-1} \leq (2j+1)2^{t-1} \\
		&\Rightarrow k+2^t \in \mathcal{I}_{t}^c,
	\end{aligned}
\end{equation*}
\begin{equation*}
	\begin{aligned}
		k \in \mathcal{I}_{t}^c & \Rightarrow (j-1)2^t+1+2^{t-1} \leq k+2^{t-1} \leq (2j-1)2^{t-1}+2^{t-1} \\
		&\Rightarrow (2j-1)2^{t-1}+1 \leq k+2^{t-1} \leq 2j2^{t-1} \\
		&\Rightarrow k+2^t \in \mathcal{I}_{t}.
	\end{aligned}
\end{equation*}
As can be seen, only one of the two segments $k$ and $k+2^{t-1}$ is involved in the $t$-th S-PC. Therefore, for the $t$-th S-PC, segments $k$ and $k+2^{t-1}$ are inequivalent, which proves the case of $\nu=1$. Then, since segments $k+2^{t-1}$ and $k+2^{t-1}+\mu2^t$ are equivalent based on the above proof, we can conclude that segment $k$ is also inequivalent to segment $k+(2\nu-1)2^{t-1}$ with $\nu \in \{\lceil -k/2^t+1/2 \rceil,\ldots,\lfloor (2^d-k)/2^t+1/2 \rfloor\}$. This thus completes the proof.

{
\bibliographystyle{IEEEtran}
\bibliography{IEEEabrv,mybibfile}
}

\end{document}